\documentclass[11pt]{article}

\newcommand{\ifabs}[2]{#2}

\ifabs{}{}

\usepackage{comment}
\usepackage{cite}
\usepackage{fullpage}
\usepackage{color}
\usepackage{epic}
\usepackage{eepic}
\usepackage{epsfig}
\usepackage{xspace}
\usepackage{graphicx}
\usepackage{amsmath}
\usepackage{amssymb}
\usepackage[ruled,vlined]{algorithm2e}
\ifabs{
\usepackage{times}
}{}

\ifabs{ \excludecomment{fullproof}
\newenvironment{proofsketch}{\noindent{\bf Sketch of Proof: }}{\nopagebreak\rule{1 ex}{0.8 em}\medskip}
\includecomment{abs}
\excludecomment{fullpaper}

}{

\excludecomment{proofsketch} \excludecomment{abs}
\includecomment{fullpaper}
}
 
\excludecomment{suppress}

\newcommand{\todo}[1]{\typeout{TODO: \the\inputlineno: #1}\textbf{[[[ #1 ]]]}}
\renewcommand{\todo}[1]{\typeout{TODO: \the\inputlineno: #1}}

\newcommand{\concept}[1]{\emph{#1}}

\newtheorem{theorem}{Theorem}
\newtheorem{lemma}[theorem]{Lemma}

\newtheorem{definition}[theorem]{Definition}
\newtheorem{definitions}[theorem]{Definitions}

\newcommand{\newloglike}[2]{\newcommand{#1}{\mathop{\rm #2}\nolimits}}
\newloglike{\sgn}{sgn}

\newcommand{\nul}[1]{{\it et al.\/}}

\newenvironment{proof}{\noindent{\bf Proof: }}{\nopagebreak\rule{1 ex}{0.8 em}\medskip}

\newcommand{\deriv}[2]{\frac{\operatorname{d}{#1}}{\operatorname{d}{#2}}}
\newcommand{\pderiv}[2]{\frac{\partial{#1}}{\partial{#2}}}
\newcommand{\blue}{\mathtt{blue}}

\newcommand{\true}{\mathtt{true}}
\newcommand{\false}{\mathtt{false}}

\begin{document}

\ifabs{
\title{Approximate Counting via Correlation Decay in Spin Systems \\ {\Large(Extended Abstract)} }
}{
\title{Approximate Counting via Correlation Decay in Spin Systems }
}
\author{
Liang Li~\thanks{This work was done when these authors visited Microsoft Research Asia.}\\
Peking University\\
\texttt{liang.li@pku.edu.cn}
\and
Pinyan Lu\\
Microsoft Research Asia\\
\texttt{pinyanl@microsoft.com}\and
Yitong Yin\thanks{Supported by the National Science Foundation of China under Grant No. 61003023
and No. 61021062.}~\footnotemark[1]\\
Nanjing University\\
\texttt{yinyt@nju.edu.cn}
}
\date{}

\maketitle

\begin{abstract}
We give the first deterministic fully polynomial-time approximation scheme (FPTAS) for computing the partition function of a two-state spin system on an arbitrary graph, when the parameters of the system satisfy the uniqueness condition on infinite regular trees. This condition is of physical significance and is believed to be the right boundary between approximable and inapproximable.

The FPTAS is based on the correlation decay technique introduced by Bandyopadhyay and Gamarnik [SODA 06] and Weitz [STOC 06].
 The classic correlation decay is defined with respect to graph distance. Although this definition has natural physical meanings, it does not directly support an FPTAS for systems on arbitrary graphs, because for graphs with unbounded degrees, the local computation that provides a desirable precision by correlation decay may take super-polynomial time. We introduce a notion of \concept{computationally efficient correlation decay}, in which the correlation decay is measured in a refined metric instead of graph distance. We use a potential method to analyze the amortized behavior of this correlation decay and establish a correlation decay that guarantees an inverse-polynomial precision by polynomial-time local computation. This gives us an FPTAS for spin systems on arbitrary graphs. This new notion of correlation decay properly reflects the algorithmic aspect of the spin systems, and may be used for designing FPTAS for other counting problems.

\end{abstract}

\ifabs{ \setcounter{page}{0} \thispagestyle{empty} \vfill
\pagebreak }{
}

\section{Introduction}
Spin systems are well studied in Statistical Physics. We focus on two-state spin systems.
An instance of a spin system is a graph $G=(V,E)$. A configuration $\sigma: V\rightarrow \{0, 1\}$ assigns every vertex one of the two states. We shall refer the two states as blue and green.
The contributions of local interactions between adjacent vertices are quantified by a
matrix 
\[A=\begin{bmatrix} A_{0,0} & A_{0,1} \\ A_{1,0} & A_{1,1} \end{bmatrix}=\begin{bmatrix} \beta & 1 \\ 1 & \gamma \end{bmatrix},
\]
where $\beta, \gamma \geq 0$.
The weight of an assignment is the production of contributions of all local interactions and the partition function $Z_A(G)$ of a system is the summation of the weights
over all  possible assignments. Formally,
\[ Z_A(G)=\sum_{\sigma \in 2^{V}} \prod_{(u,v) \in E} A_{\sigma(u), \sigma(v)}.\]
Although originated from Statistical Physics, the spin system is also accepted in Computer Science as a framework for counting problems. Considering the two very well studied frameworks, the weighted Constraint Satisfaction Problems (\#CSP)~\cite{CreignouH96,BulatovD03,weightedCSP,Bulatov08,STOC09,Dyer-Rich,ccl-csp} and Graph Homomorphisms~\cite{DyerG00,BulatovG05,acyclic,GGJT09,Homomorphisms,CaiChen}, the two-state spin systems can be viewed as the most basic setting in these frameworks: A Boolean \#CSP problem with one symmetric binary relation; or Graph Homomorphisms to graph with two vertices. Many natural combinatorial problems can be formulated as two-state spin systems. For example, with $\beta=0$ and $\gamma=1$, $Z_A(G)$ is the number of independent sets (or vertex covers) of the graph $G$.

Given a matrix $A$, it is a computational problem to compute $Z_A(G)$ where graph $G$ is given as input.
We want to characterize the computational complexity of computing $Z_A(G)$ in terms of $\beta$ and $\gamma$.
For \emph{exact} computation of $Z_A(G)$, polynomial time algorithms are known only for the very restricted settings that $\beta \gamma =1$ or $(\beta,\gamma)= (0,0)$, and for all other settings the problem is proved to be \#P-Hard \cite{BulatovG05}.
We consider the approximation of $Z_A(G)$, with the fully polynomial-time approximation schemes (FPTAS) and its randomized relaxation the fully polynomial-time randomized approximation schemes (FPRAS).


In a seminal paper~\cite{JS93}, Jerrum and Sinclair gave an FPRAS when $\beta=\gamma >1$, which was further extended to the entire region $\beta \gamma >1$~\cite{GJP03}.
For $0\leq \beta, \gamma \leq 1$ except that  $(\beta,\gamma)= (0,0)$ or $(1,1)$, Goldberg,  Jerrum and Paterson prove that the problem do not admit an  FPRAS unless NP$=$RP~\cite{GJP03}.
For the other values of the parameters, namely, $0\le \beta<1<\gamma<\frac{1}{\beta}$ or symmetrically $0\le \gamma<1<\beta<\frac{1}{\gamma}$, the approximability of $Z_A(G)$ is not very well understood.
It was shown in~\cite{GJP03} that by coupling a simple heat-bath random walk, there exists an additional region of $\beta$ and $\gamma$ which admit some FPRAS. The true characterization of approximability  is still left open.

Within this unknown region, there lies a critical curve with physical significance, called the uniqueness threshold.
The phase transition of Gibbs measure occurs at this threshold curve.
Such statistical physics phase transitions are believed to coincide with the transitions of computational complexity. However, there are only very few examples where the connection is rigorously proved.
One example is the hardcore (counting independent set) model. It was conjectured in \cite{inapp_MWW09} by  Mossel, Weitz and Wormald, and settled in a line of works by Dyer, Frieze and Jerrum \cite{IS_DFJ02},  Weitz \cite{Weitz06},  Sly \cite{Sly10}, and very recently Galanis, Ge, {\v S}tefankovi{\v c}, Vigoda and Yang \cite{GGSVY11}
that in the hardcore model the uniqueness threshold essentially characterizes the approximability of the partition function. It will be very interesting to observe the similar transition in spin systems.

\subsection{Main results}
We extend the approximable region (in terms of $\beta$ and $\gamma$) of $Z_A(G)$ to the uniqueness threshold in two-state spin systems, which is believed to be the right boundary between approximable and inapproximable.
Specifically, we formulate a criterion for $\beta$ and $\gamma$ such that there is a unique Gibbs measure on all infinite regular trees\footnote{Technically, there is a small integrality gap caused by the continuous generalization of the condition. The formal statement is given in the following section.}, and prove that there is an FPTAS for computing $Z_A(G)$ when this uniqueness condition is satisfied. This improves the approximable boundary (dashed lines in Figure 1) provided by the heat-bath random walk in \cite{GJP03}. Moreover, the algorithm is deterministic.

The FPTAS is based on the correlation decay technique first used in \cite{Weitz06,BG08} for approximate counting. We elaborate a bit on the ideas. A spin system induces a natural probability distribution over all configurations called the Gibbs measure where the probability of a configuration is proportional to its weight. Due to a standard self-reduction procedure, computing $Z_A(G)$ is reduced to computing the marginal distribution of the state of one vertex, which is made plausible by Weitz in \cite{Weitz06} with the self-avoiding walk (SAW) tree construction.
For efficiency of computation, the marginal distribution of a vertex is estimated using only a local neighborhood around the vertex. To justify the precision of the estimation, we show that far-away vertices have little influence on the marginal distribution. This is done by analyzing the rate with which the correlation between two vertices decays as they are far away from each other.

The correlation decay by itself is a phenomenon of physical significance. One of our main discoveries is that two-state spin systems on any graphs have exponential correlation decay when the above uniqueness condition is satisfied.

\ifabs{
\begin{figure}
\centering\includegraphics[width=200pt]{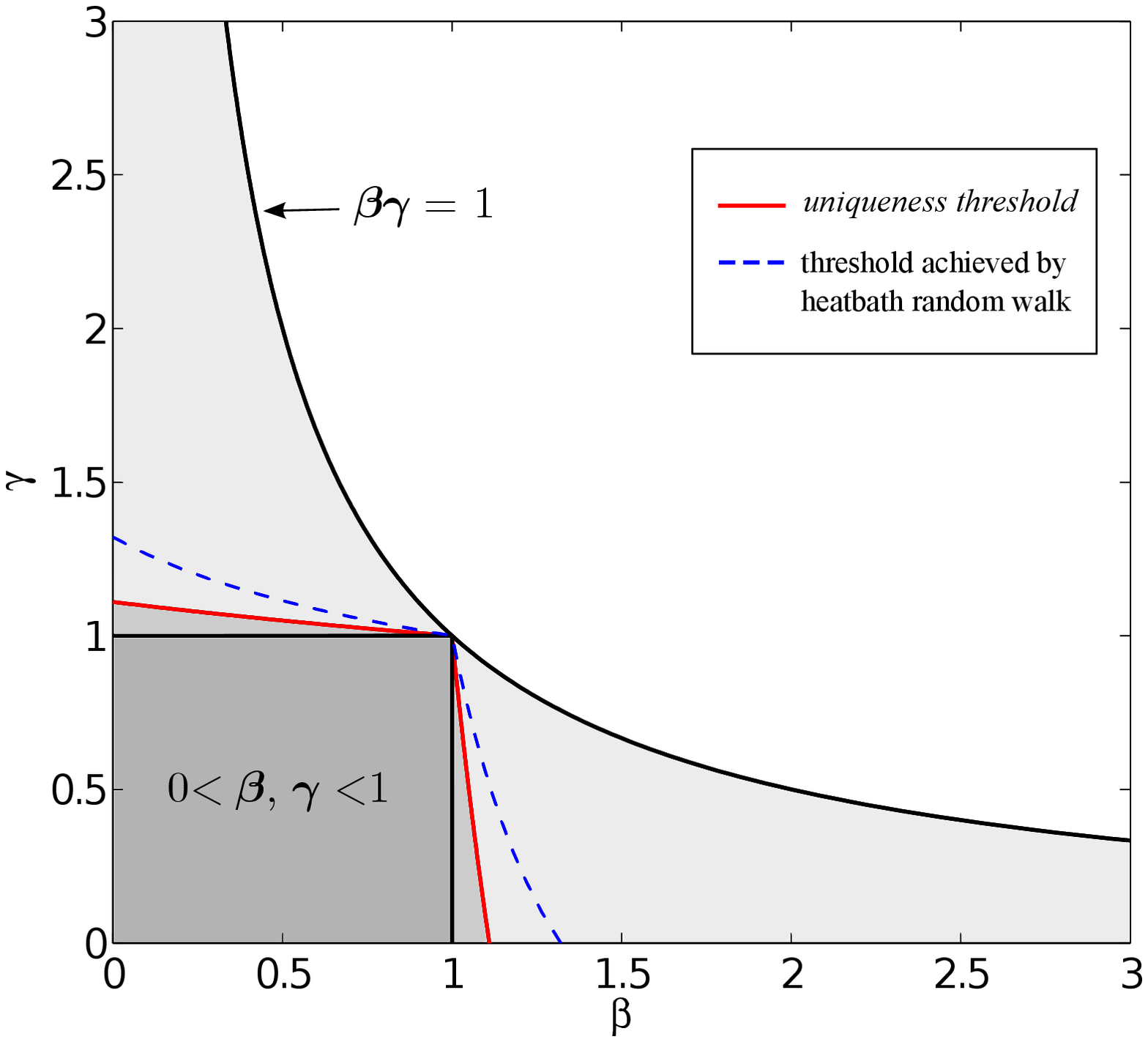}
\caption{Our FPTAS works for the area between the uniqueness threshold and the line $\beta\gamma=1$, and the heat-bath random walk in \cite{GJP03} works between the dashed line and $\beta\gamma=1$.}
\end{figure}
}{
\begin{figure}\label{fig:beta-gamma}
\centering\includegraphics[width=280pt]{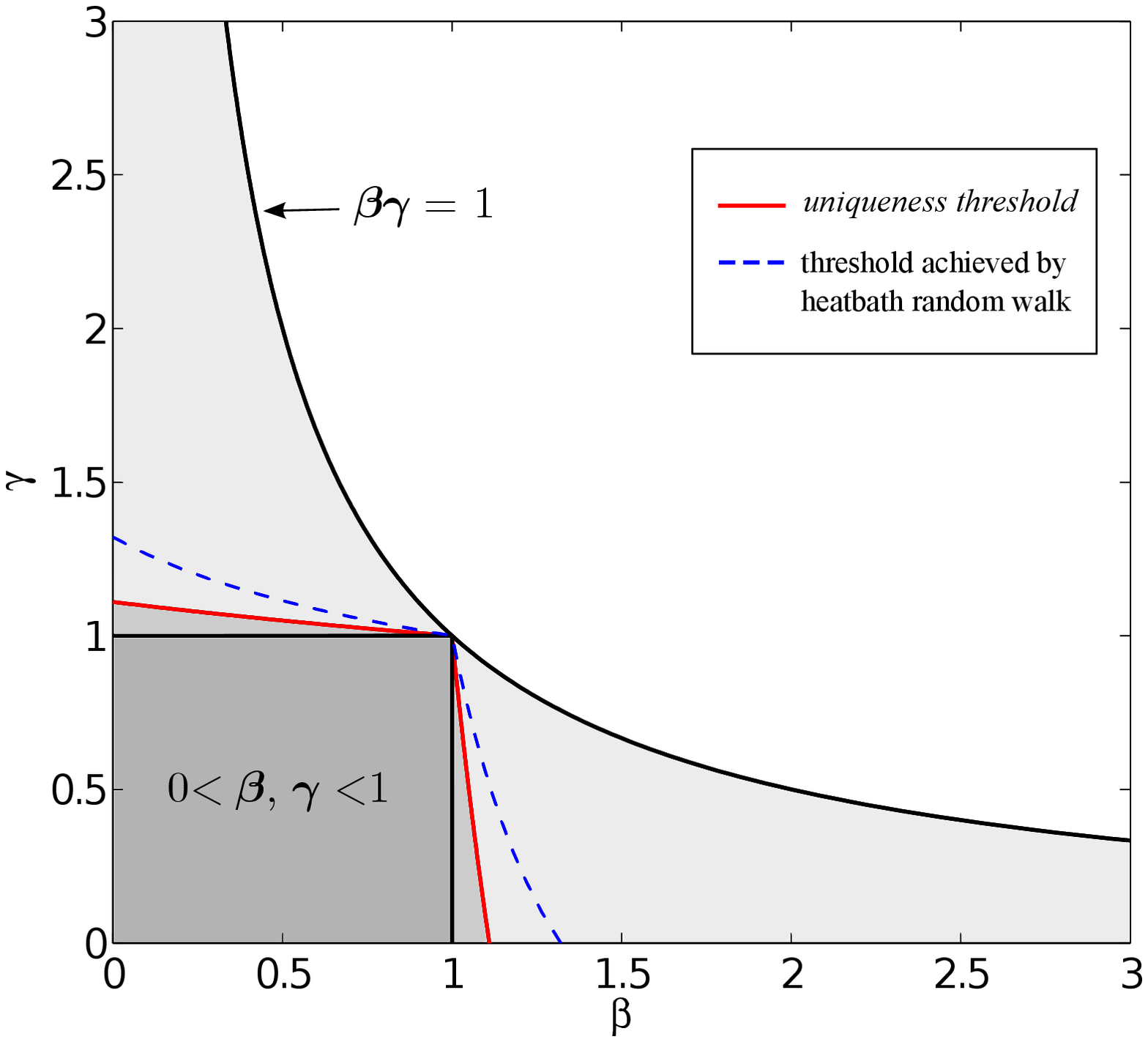}
\caption{Our FPTAS works for the region between the critical curve of the uniqueness threshold and the curve $\beta\gamma=1$. The heat-bath random walk in \cite{GJP03} works for the region between the dashed line and $\beta\gamma=1$.}
\end{figure}
}

\subsection{Technical contributions}
The technique of using correlation decay to design FPTAS for partition functions is developed in the hardcore model. We introduce several new ideas to adapt the challenges arising from spin systems. We believe these challenges are typical in counting problems, and the new ideas will make the correlation decay technique more applicable for approximate counting.

\begin{enumerate}
\item The correlation decay technique used in \cite{Weitz06} relies on a monotonicity property specific to the hardcore model. Correlation decays in graphs are reduced via this monotonicity to the decays in infinite regular trees, while the later have solvable phase transition thresholds.
It was already observed in \cite{Weitz06} that such monotonicity may not generally hold for other models. Indeed, it does not hold for spin systems. We develop a more general method which does not rely on monotonicity: We directly compute the correlation decay in arbitrary trees (and as a result in arbitrary graphs via the SAW tree reduction), and use the potential method to analyze the amortized behavior of correlation decay.

\item To have an FPTAS, the marginal distribution of a single vertex should be approximable up to certain precision from a local neighborhood of polynomial size. The classic correlation decay is measured with respect to graph distance.
The local neighborhoods in this sense are balls in the graph metric.  A SAW tree enumerates all paths originating from a vertex.
For graphs of unbounded degrees, the SAW tree transformation may have the balls offering desirable precisions explode to super-polynomial sizes.

We introduce the notion of \concept{computationally efficient correlation decay}. Correlation decay is now measured in a refined metric, which has the advantage that a desirable precision is achievable by a ball (in the new metric) of polynomial size even after the SAW tree transformation. We prove an exponential correlation decay in this new metric when the uniqueness is satisfied. As a result, we have an FPTAS for \emph{arbitrary} graphs as long as the uniqueness condition holds.
\end{enumerate}

\subsection{Related works}
The approximation for partition function has been extensively studied with both positive~\cite{JS93,app_JSV04,app_GJ11,app_DJV01,IS_DG00,col_Jerrum95,col_Vigoda99} and negative results~\cite{inapp_GJ10,inapp_BCFKTVV99,inapp_MWW09,inapp_GJ09,inapp_CDF01,inapp_GT04,inapp_BR04,inapp_GR07}. Some special problems in these framework are well studied combinatorial problems, e.g.~counting independent sets~\cite{IS_DFJ02,IS_DG00,IS_LV97} and graph coloring~\cite{MSW07,col_Jerrum95,col_HV03,col_DFFV06,col_DFHV04,col_Hayes03,col_DF03,col_HVV07,col_HV05,col_DGM02,col_Molloy04,col_Vigoda99,col_BDK08,col_GJK10}.
Some dichotomies (or trichotomies) of complexity for approximate counting CSP were also obtained~\cite{dicho_DGJ10,dicho_DGJR10,GGSVY11,Sly10}.
Almost all known approximation counting algorithms are based on random sampling~\cite{samp_JVV86,samp_DGJ04}, usually through the famous Markov Chain Monte Carlo (MCMC) method~\cite{MC_DG99,MC_JA96}.
There are very few deterministic approximation algorithms for any counting problems. Some notable examples include~\cite{BG08,GK07,BGKNT07,GKM10,SVV10}.

In a very recent work \cite{sinclair2011approximation}, Sinclair, Srivastava, and Thurley give an FPTAS using correlation decay for the two-state spin systems on bounded degree graphs. They allow the two-state spin systems to have an external field, and the uniqueness thresholds they used are defined with respect to specific maximum degrees.


\section{Definitions and Statements of Results}\label{section-definition}
A spin system is described by a graph $G=(V,E)$. A \concept{configuration} of the system is one of the $2^{|V|}$ possible
assignments $\sigma: V\rightarrow \{0, 1\}$ of states to vertices. We also use two colors \emph{blue} and \emph{green} to denote these two states.
Let $A=\begin{bmatrix} A_{0,0} & A_{0,1} \\ A_{1,0} & A_{1,1} \end{bmatrix}=\begin{bmatrix} \beta & 1 \\ 1 & \gamma \end{bmatrix}$, where $\beta, \gamma \geq 0$.
The \concept{Gibbs measure} is a distribution over all configurations defined by
\begin{align*}
\mu(\sigma)=\frac{1}{Z_A(G) }\prod_{(u,v) \in E} A_{\sigma(u), \sigma(v)}.
\end{align*}
The normalization factor $Z_A(G)=\sum_{\sigma \in 2^{V}} \prod_{(u,v) \in E} A_{\sigma(u), \sigma(v)}$
is called the \concept{partition function}.

From this distribution, we can define the marginal probability $p_v$ of $v$ to be colored blue. 
Let $\sigma_{\Lambda}$ be a configuration defined on vertices in $\Lambda\subset V$.  We call vertices $v\in\Lambda$ \concept{fixed} vertices, and $v\not\in\Lambda$ \concept{free} vertices. We use $p_v^{\sigma_{\Lambda}}$ to denote the marginal probability of $v$ to be colored blue conditioned on the configuration of $\Lambda$ being fixed as $\sigma_{\Lambda}$.

\begin{definition}\label{definition-correlation-decay}
A spin system on a family of graphs is said to have exponential correlation decay if for any graph $G=(V,E)$ in the family, any $v\in V,\Lambda\subset V$ and $\sigma_{\Lambda},\tau_{\Lambda}\in\{0,1\}^{\Lambda}$,
\begin{align*}
|p_v^{\sigma_{\Lambda}}-p_v^{\tau_{\Lambda}}|\leq \exp(-\Omega(\mathrm{dist}(v,\Delta))).
\end{align*}
 where $\Delta\subset\Lambda$ is the subset on which $\sigma_{\Lambda}$ and $\tau_{\Lambda}$ differ, and $\mathrm{dist}(v,\Delta)$ is the shortest distance from $v$ to any vertex in $\Delta$.
\end{definition}

This definition is equivalent to the ``strong spatial mixing'' in \cite{Weitz06} with an exponential rate. It is stronger than the standard notion of exponential correlation decay in Statistical Physics \cite{Dob70}, where the decay is measured with respect to $\mathrm{dist}(v,\Lambda)$ instead of $\mathrm{dist}(v,\Delta)$.

The marginal probability $p_v^{\sigma_{\Lambda}}$ in a tree can be computed by the following recursion.
Let $T$ be a tree rooted by $v$. We denote $R_T^{\sigma_{\Lambda}}$ as the ratio of the probabilities that root $v$ is blue and green, respectively, when imposing the condition $\sigma_{\Lambda}$. Formally, $R_T^{\sigma_{\Lambda}}=\frac{p_v^{\sigma_{\Lambda}}}{1-p_v^{\sigma_{\Lambda}}}$ (when $p_v^{\sigma_{\Lambda}}=1$, let $R_v^{\sigma_{\Lambda}}=\infty$ by convention). 
Suppose that the root of $T$ has $d$ children. Let $T_i$ be the subtree rooted by the $i$-th child of the root.  The distributions on distinct subtrees are independent.
A calculation then gives that
\begin{equation}
R_T^{\sigma_{\Lambda}} = \prod_{i=1}^d \frac{\beta R_{T_i}^{\sigma_{\Lambda}}+1 }{R_{T_i}^{\sigma_{\Lambda}} +\gamma}. \label{eq:recursion}
\end{equation}
%
%

It is of physical significance to study the Gibbs measures on infinite $(d+1)$-regular trees $\widehat{\mathbb{T}}^d$ \cite{Georgii88}. In $\widehat{\mathbb{T}}^d$, the recursion is of a symmetric form $f(x)=\left(\frac{\beta x + 1}{x + \gamma}\right)^d$. There may be more than one Gibbs measures on infinite graphs. We say that the system has the \concept{uniqueness}  if there is exact one Gibbs measure. Let  $\hat{x}=f(\hat{x})$ be the fixed point of $f(x)$. It is known \cite{kelly1985stochastic,MSW07} that the spin system on $\widehat{\mathbb{T}}^d$ undergoes a phase transition at $|f'(\hat{x})|=1$ with uniqueness when $|f'(\hat{x})|=\frac{d(1-\beta\gamma)(\beta \hat{x}+1)^{d-1}}{(\hat{x}+\gamma)^{d+1}}\le 1$.
This motivates the following definition
\begin{align*}
\Gamma(\beta)&=\inf\left\{\gamma\ge1 \Bigm{|} \forall d\ge 1, \,\frac{d(1-\beta\gamma)(\beta \hat{x}+1)^{d-1}}{(\hat{x}+\gamma)^{d+1}}\le1\right\}.
\end{align*}
For a fixed $0\le\beta<1$, the $\Gamma(\beta)$ gives the boundary that all infinite regular trees $\widehat{\mathbb{T}}^d$ exhibit uniqueness when $\Gamma(\beta)\le\gamma\le\frac{1}{\beta}$. We call $\Gamma(\beta)$ the \concept{uniqueness threshold}\ifabs{\footnote{For technical reasons, we treat $d$ as real numbers. This will make $\Gamma(\beta)$ only slightly greater than the uniqueness threshold $\Gamma^*(\beta)$ realized by integer $d$. For example, when $\beta=0$, $\Gamma(\beta)\approx 1.1101715$ and $\Gamma^*(\beta)\approx1.1101714$}}{}.
\ifabs{In the full version of the paper, we will show that $1<\Gamma(\beta)<\frac{1}{\beta}$ for any $0\le\beta<1$.}{}
Indeed, for any $d\ge 1$, there is a critical $\Gamma_d(\beta)$ such that $\widehat{\mathbb{T}}^d$ exhibits uniqueness when $\Gamma_d(\beta)<\gamma<\frac{1}{\beta}$. Furthermore, there is a finite crucial $D>1$ such that $\Gamma_D(\beta)=\Gamma(\beta)$. That is, $\widehat{\mathbb{T}}^D$ has the highest uniqueness threshold $\Gamma(\beta)$ among all $\widehat{\mathbb{T}}^d$.

\ifabs{}{We remark that for technical reasons, we treat $d$ as real numbers thus $\Gamma(\beta)$ is slightly greater than the one defined by integer $d$s. An integer version of $\Gamma(\beta)$ is given in Section \ref{sec-beta=0}, where a slightly improved and tight analysis is given for the specially case $\beta=0$.
}

\begin{definition}
A fully polynomial-time approximation scheme (FPTAS) for $Z_A(G)$ is an algorithm that given as input an instance $G$ and an $\epsilon>0$, outputs a number $Z$ in time $\mathrm{poly}(|G|, \frac{1}{\epsilon})$ such that $(1-\epsilon)Z_A(G)\leq Z \leq (1+\epsilon)Z_A(G)$.
\end{definition}

In Definition \ref{definition-correlation-decay}, the correlation decay is measured in graph distance.
In order to support an FPTAS for graphs with unbounded degrees, we need to define the following refined metric.
\begin{definition}
Let $T$ be a rooted tree and $M\geq 2$ be a constant.  We define the \concept{$M$-based depth} $L_M(v)$ of a vertex $v$ in $T$ recursively as follows: $L_M(v)=0$ if $v$ is the root of $T$; and for every child $u$ of $v$, if $v$ has $d\ge 1$ children, $L_M(u)=L_M(v) + \lceil \log _M (d+1) \rceil$.
\end{definition}
If every vertex in $T$ has $d<M$ children, $L_M(v)$ is precisely the depth of $v$. If there are vertices having $d\ge M$ children, we actually replace every such vertex and its $d$ children with an $M$-ary tree of depth $\lceil \log _M (d+1) \rceil$, and $L_M(v)$ is the depth of $v$ in this new tree.

\begin{definition}
Let $T$ be a rooted tree and $M\geq 2$ be a constant.  Let $B_M(L)=\{v\in T\mid L_M(v)\le L\}$, called an \concept{$M$-based $L$-ball}, be the set of vertices in $T$ whose $M$-based depths are no greater than $L$; and let $B_M^*(L)$, called an \concept{$M$-based $L$-closed-ball}, be the set of vertices in $B_M(L)$ and all their children in $T$.
\end{definition}


The main technical result of the paper is the following theorem which establishes an exponential correlation decay in the refined metric when the uniqueness condition holds.

\begin{theorem}[Computationally Efficient Correlation Decay]
\label{thm:rank-decay}
Let  $0\le\beta<1$, $\beta\gamma<1$, and $\gamma>\Gamma(\beta)$. There exists a sufficiently large constant $M$ which depends only on $\beta$ and $\gamma$, such that on an arbitrary tree $T$, for any two configurations $\sigma_{\Lambda}$ and $\tau_{\Lambda}$ which differ on $\Delta \subset \Lambda$, if $B_M^*(L) \cap \Delta =\emptyset$ then
\[|R_T^{\sigma_{\Lambda}}-R_T^{\tau_{\Lambda}}| \leq \exp(-\Omega(L)).\]
\end{theorem}

The name computationally efficient correlation decay is due to the fact that $|B_M(L)|\le M^L$ in any tree, thus  an exponential decay would imply a polynomial-size $B_M(L)$ giving an inverse-polynomial precision.

Theorem \ref{thm:rank-decay} has the following implications via Weitz's self-avoiding tree construction \cite{Weitz06}.

\begin{theorem}\label{thm:correlation-decay}
Let  $0\le\beta<1$, $\beta\gamma<1$ , $\gamma>\Gamma(\beta)$. It is of exponential correlation decay for the Gibbs measure on any graph.
\end{theorem}

\begin{theorem}\label{thm:FPTAS}
Let  $0\le\beta<1$, $\beta\gamma<1$ , $\gamma>\Gamma(\beta)$. There is an FPTAS for computing the partition function $Z_A(G)$ for arbitrary graph $G$.
\end{theorem}

By symmetry, in Theorem \ref{thm:rank-decay}, \ref{thm:correlation-decay}, and \ref{thm:FPTAS}, the roles of $\beta$ and $\gamma$  can be switched.

In the Section \ref{section-algorithms}, we will show the FPTAS implied by Theorem \ref{thm:rank-decay}, followed by \ifabs{}{a formal treatment of the uniqueness threshold in Section \ref{app-Gamma-D}, and finally }the formal proof of Theorem \ref{thm:rank-decay} in Section \ref{section-analysis}.

\section{An FPTAS for the Partition Function}\label{section-algorithms}
Assuming that Theorem \ref{thm:rank-decay} is true, we show that when $0\le\beta<1$ and $\Gamma(\beta)<\gamma<\frac{1}{\beta}$, there is an FPTAS for the partition function $Z_A(G)$ for arbitrary graph $G$. The FPTAS is based on approximation of $R_{G,v}^{\sigma_\Lambda}=p_v^{\sigma_\Lambda}/(1-p_v^{\sigma_\Lambda})$, the ratio between the probabilities that $v$ is blue and green, respectively, when imposing the condition $\sigma_{\Lambda}$.

The self-avoiding walk tree is introduced by Weitz in~\cite{Weitz06} for calculating $R_{G,v}^{\sigma_\Lambda}$.
Given a graph $G=(V,E)$, we fix an arbitrary order $<$ of vertices. Originating from any vertex $v\in V$, a self-avoiding walk tree,  denoted $T_{\mathrm{SAW}}(G,v)$, is constructed as follows. Every vertex in $T_{\mathrm{SAW}}(G,v)$ corresponds to one of the walks $v_1\rightarrow v_2\rightarrow\cdots\rightarrow v_k$ in $G$ such that $v_1=v$, all edges are distinct and $v_1,\ldots, v_{k-1}$ are distinct, i.e.~the self-avoiding walks originating from $v$ and those appended with a vertex closing a cycle. The root of $T_{\mathrm{SAW}}(G,v)$ corresponds to the trivial walk $v$. The vertex $v_1$ parents $v_2$ in $T_{\mathrm{SAW}}(G,v)$, if and only if their respective walks $w_1$ and $w_2$ satisfy that $w_2=w_1\rightarrow u$ for some $u$. For a leaf of $T_{\mathrm{SAW}}(G,v)$ whose walk closes a cycle, supposed that the cycle is $u\rightarrow v_1\rightarrow\cdots v_k\rightarrow u$, fix the leaf to be blue if $v_1> v_k$ and green otherwise.
When a configuration $\sigma_{\Lambda}$ is imposed on $\Lambda\subset V$ of the original graph $G$, for any vertex of $T_{\mathrm{SAW}}(G,v)$ whose corresponding walk ends at a $u\in\Lambda$, the color of the vertex is fixed to be $\sigma_\Lambda(u)$. We abuse the notation and denote the resulting configuration on $T_{\mathrm{SAW}}(G,v)$ by $\sigma_\Lambda$ as well.

This novel tree construction has the advantage that the probabilities are exactly the same in both the original spin system and the constructed tree.
\begin{theorem}[Weitz~\cite{Weitz06}]\label{theorem-T-SAW}
Let $T=T_{\mathrm{SAW}}(G,v)$.
It holds that $R^{\sigma_\Lambda}_{G,v}=R^{\sigma_\Lambda}_{T}$ 
\end{theorem}

Due to \eqref{eq:recursion}, in a tree $T$, the following recursion holds for $R_T^{\sigma_{\Lambda}}$:
\begin{align*}
R_T^{\sigma_{\Lambda}} 
&= \prod_{i=1}^d \frac{\beta R_{T_i}^{\sigma_\Lambda}+1 }{R_{T_i}^{\sigma_\Lambda} +\gamma}.
\end{align*}
The base case is either when the current root $v\in \Lambda$, i.e.~$v$'s color is fixed, in which case $R_T^{\sigma_{\Lambda}} = \infty$  or $R_T^{\sigma_{\Lambda}} =0$
(depending on whether $v$ is fixed to be blue or green), or when $v$ is free and has no children, in which case $R_T^{\sigma_{\Lambda}}=1$ (this is consistent with the recursion since the outcome of an empty product is 1 by convention).

For $\beta\gamma<1$, the recursion is monotonically decreasing with respect to every $R_{T_i}^{\sigma_\Lambda}$.
An upper (lower) bound of $R_T^{\sigma_{\Lambda}}$ can be computed by replacing $R_{T_i}^{\sigma_\Lambda}$ in the recursion by their respective lower (upper) bounds.
Algorithm \ref{alg:R-bounds} computes the lower or upper bound of $R_T^{\sigma_\Lambda}$ up to vertices in  $M$-based $L$-closed ball $B_M^*(L)$. For the vertices outside $B_M^*(L)$, it uses the trivial bounds $0\le R\le \infty$.

\begin{algorithm}\label{alg:R-bounds}
\caption{Estimate $R_T^{\sigma_\Lambda}$ based on $B_M^*(L)$}
$R(T_v, \sigma_\Lambda,L,d_{\mathrm{parent}},lb)$:\\
\KwIn{Rooted tree $T_v$; configuration $\sigma_\Lambda$; $M$-based depth $L$; parent degree $d_{\mathrm{parent}}$; Boolean indicator $lb$ of lower bound.}
\KwOut{Lower (or upper) bound of $R_T^{\sigma_\Lambda}$ computed from vertices in $B_M^*(L)$.}
\Begin{
Suppose root $v$ has $d$ children and let $T_i$ be the subtree rooted by the $i$-th child\;
\uIf{$v\in\Lambda$}{
\lIf{$\sigma_\Lambda(v)=\blue$}{\Return $\infty$\;}
\lElse{ \Return $0$\;}
}\uElseIf{$L< 0$}{
\lIf{$lb=\true$}{\Return $0$\;}
\lElse{ \Return $\infty$\;}
}\Else{
$L'\longleftarrow L-\lceil\log_M(d_{\mathrm{parent}}+1)\rceil$\;
\Return $\prod_{i=1}^d\frac{\beta R(T_i,\sigma_\Lambda,L',d,\neg lb)+1}{R(T_i,\sigma_\Lambda,L',d,\neg lb)+\gamma}$\;
}
}
\end{algorithm}

Due to the monotonicity of the recursion, it holds that
\[R(T,\sigma_\Lambda,L,0,\true)\le R_T^{\sigma_\Lambda}\le R(T,\sigma_\Lambda,L,0,\false).
\]
Note that the naive lower bound $0$ (or the upper bound $\infty$) of $R$ for a vertex outside $B_M^*(L)$ can be achieved by fixing the vertex to be green (or blue). Denote by $\tau_0$ and $\tau_1$ the configurations achieving the lower and upper bounds respectively. 
It is easy to see that $\tau_0=\tau_1=\sigma_\Lambda$ in $B_M^*(L)$. Then due to Theorem \ref{thm:rank-decay}, there is a constant $\alpha<1$ such that 
\begin{align*}
&\quad\,\, |R(T,\sigma_\Lambda,L,0,\false)-R(T,\sigma_\Lambda,L,0,\true)|\\
&=|R_T^{\tau_1}-R_T^{\tau_0}|\\
&=O(\alpha^L).
\end{align*}

To compute $R_{G,v}^{\sigma_{\Lambda}}$ for an arbitrary graph $G$, we first construct the $B_M^*(L)$ of $T=T_{\mathrm{SAW}}(G,v)$, and run Algorithm \ref{alg:R-bounds}. Due to Theorem \ref{theorem-T-SAW}, $R_{G,v}^{\sigma_{\Lambda}}=R_T^{\sigma_\Lambda}$, thus it returns $R_0$ and $R_1$ such that $R_0\le R_{G,v}^{\sigma_{\Lambda}}\le R_1$ and $R_1-R_0= O(\alpha^L)$. Since $p_v^{\sigma_\Lambda}=R_{G,v}^{\sigma_{\Lambda}}/(1+R_{G,v}^{\sigma_{\Lambda}})$, we can output $p_0=\frac{R_0}{R_0+1}$ and $p_1=\frac{R_1}{R_1+1}$ so that $p_0\le p_v^{\sigma_\Lambda}\le p_1$ and $p_1-p_0=\frac{R_1}{R_1+1}-\frac{R_0}{R_0+1}\le R_1-R_0=O(\alpha^L)$.

The running time of this algorithm relies on the size of $B_M^*(L)$ in $T_{\mathrm{SAW}}(G,v)$.
The maximum degree of $T_{\mathrm{SAW}}(G,v)$ is bounded by the maximum degree of $G$, which is trivially bounded by $n$, thus $|B_M^*(L)|\le n|B_M(L)|\le nM^L$. The running time of the algorithm is $O(|B_M^*(L)|)=O(nM^L)$.

By setting $L=\log_\alpha\epsilon$, we can approximate $1-p_{v}^{\sigma_\Lambda}$ within absolute error $O(\epsilon)$ in time $O(n \cdot \mathrm{poly}(\frac{1}{\epsilon}))$.
For $\beta<1<\gamma$, it holds that  $0<R_{G,v}^{\sigma_{\Lambda}}<1$  for free $v$ thus $1-p_{v}^{\sigma_\Lambda}>\frac{1}{2}$,  therefore the above procedure approximates $(1-p_{v}^{\sigma_\Lambda})$ within factor $(1\pm O(\epsilon))$. We have an FPTAS for $(1-p_{v}^{\sigma_\Lambda})$.

The partition function $Z_A(G)$ can be computed from $p_v^{\sigma_\Lambda}$ by the following standard routine.  Let $v_1,\ldots,v_n$ enumerate the vertices in $G$, and let $\sigma_i, i=0,1,\ldots,n$, be the configurations fixing the first $i$ vertices $v_1,\ldots,v_{i}$ to be green, where $\sigma_0$ means all vertices are free. The probability measure of $\sigma_{n}$ (all green) can be computed as
\begin{align*}
\mu(\sigma_{n})
&=\prod_{i=1}^n \Pr[v_i\text{ is green }\mid v_1,\ldots,v_{i-1}\text{ are green}]\\
&=\prod_{i=1}^n(1-p_{v_i}^{\sigma_{i-1}}).
\end{align*}
On the other hand,  it is easy to see that $\mu(\sigma_{n})=\frac{\gamma^{|E|}}{Z_A(G)}$ by definition of $\mu$.
Thus
\begin{align*}
Z_A(G)&=\frac{\gamma^{|E|}}{\mu(\sigma_{n})}=\frac{\gamma^{|E|}}{\prod_{i=1}^n(1-p_{v_i}^{\sigma_{i-1}})}.
\end{align*}
Notice that $\gamma^{|E|}>1$. Therefore, an FPTAS for $(1-p_v^{\sigma_\Lambda})$ implies an FPTAS for $Z_A(G)$.

\begin{suppress}
\begin{algorithm}
\caption{Construct the SAW tree $T_{saw}(G,v)$}
\KwIn{Graph $G=(V,E)$; an ordering $\prec_u$ on the neighboring edges of each vertex $u\in V$; an originating vertex $v\in V$.}
\KwOut{$T_{saw}(G,v)$}
\Begin{
    start from $v$, form SAW tree $\widetilde{T_{saw}}(G,v)$ and include all the vertices closing a cycle as the tree leaves\;
    \ForEach {$ul\in Leaf(\widetilde{T_{saw}}(G,v))$ closing a cycle $u\rightarrow u_1\rightarrow...\rightarrow u_m\rightarrow u$ in $G$}{
        \lIf{$(u,u_1)\prec_u (u,u_m)$}{set $ul$ to be blue\;}
        \lElse {set $ul$ to be green\;}
    }

}
\end{algorithm}
\end{suppress} 

\ifabs{}{
\section{The uniqueness threshold}\label{app-Gamma-D}
In this section, we formally define the uniqueness threshold $\Gamma(\beta)$ and the critical $D$. We also prove several propositions regarding these quantities which are useful for the analysis of the correlation decay.
\begin{definition}
Let $0\le\beta<1$ be a fixed parameter. Suppose that $1\le\gamma<\frac{1}{\beta}$ and $d\ge1$. Let $x(\gamma,d)$ be the  positive solution of
\begin{align}
x &=\left(\frac{\beta x+1}{x+\gamma}\right)^d. \label{eq:fix-point-1}
\end{align}
\end{definition}
Define that $f(x)=\left(\frac{\beta x+1}{x+\gamma}\right)^d$. Then $x(\gamma,d)$ is the positive fixed point of $f(x)$. For $\gamma<\frac{1}{\beta}$, $f(x)=\left(\beta+\frac{1-\beta\gamma}{x+\gamma}\right)^d$ is continuous and strictly decreasing over $x\in[0,\infty)$, and it holds that $f(0)=\frac{1}{\gamma^d}>0$ and $f(1)=\left(\frac{1+\beta}{1+\gamma}\right)^d<\frac{1}{\gamma^d}\le1$, thus $f(x)$ has a unique fixed point over $x\in(0,1)$. Therefore, for  $1\le\gamma<\frac{1}{\beta}$ and $d\ge1$, $x(\gamma,d)$ is well defined and $x(\gamma,d)\in(0,1)$.

\begin{definition}
Let
\begin{align*}
\Gamma(\beta)&=\inf\left\{\gamma\ge1 \Bigm{|} \forall d\ge 1, \,\frac{d(1-\beta\gamma)(\beta x(\gamma,d)+1)^{d-1}}{(x(\gamma,d)+\gamma)^{d+1}}\le1\right\}.
\end{align*}
We write $\Gamma=\Gamma(\beta)$ for short if no ambiguity is caused.
\end{definition}

Note that $\Gamma$ can be equivalently defined as
\begin{align*}
\Gamma &=\inf\left\{\gamma\ge1 \Bigm{|} \forall d\ge 1, \,\frac{d(1-\beta\gamma)x(\gamma,d)}{\left(\beta x(\gamma,d)+1\right)\left(x(\gamma,d)+\gamma\right)}\le1\right\},
\end{align*}
because $x(\gamma,d)$ satisfies \eqref{eq:fix-point-1}.

The following lemma states that for $0\le\beta<1$, $\Gamma(\beta)$ is well-defined and nontrivial.
\begin{lemma}\label{lemma-well-define-gamma}
For $0\le\beta<1$, it holds that $1<\Gamma(\beta)<\frac{1}{\beta}$.
\end{lemma}
\begin{proof}
We first show that $\Gamma>1$. It is sufficient to show that if $\gamma\le 1$ then there exists a $d$ such that $\frac{d(1-\beta\gamma)x}{(\beta x+1)(x+\gamma)}>1$, where $x$ satisfies that $x=\left(\frac{\beta x+1}{x+\gamma}\right)^d$.

By contradiction, suppose that $\gamma\le 1$ and  for all $d\ge 1$,  $\frac{d(1-\beta\gamma)x}{(\beta x+1)(x+\gamma)}\le1$ where  $x$ satisfies that $x=\left(\frac{\beta x+1}{x+\gamma}\right)^d$. Then,
\begin{align*}
1\ge \frac{d(1-\beta\gamma)x}{(\beta x+1)(x+\gamma)} = \frac{d(1-\beta\gamma)}{\beta x+\frac{\gamma}{x}+(1+\beta\gamma)}\ge  \frac{d(1-\beta\gamma)}{\beta x+\frac{\gamma}{x}+2}.
\end{align*}

Specifically, suppose that $d$ is sufficiently large so the followings hold
\begin{align*}
\beta^d\exp\left(\frac{d}{(1-\beta\gamma)d-3}\right)
&<\frac{d(1-\beta\gamma)-3}{\beta}, \text{ and}\\
\exp\left(-\frac{\gamma d}{d(1-\beta\gamma)-3}\right)
&>\frac{\gamma}{d(1-\beta\gamma)-3}.
\end{align*}

\begin{list}{}{}
\item[Case.1:] $x\ge\gamma$. Then $\frac{\gamma}{x}\le 1$. Thus,
\begin{align*}
1\ge  \frac{d(1-\beta\gamma)}{\beta x+\frac{\gamma}{x}+2} \ge \frac{d(1-\beta\gamma)}{\beta x+3},
\end{align*}
which implies that $x\ge \frac{d(1-\beta\gamma)-3}{\beta}$. On the other hand,
\begin{align*}
x &=\left(\frac{\beta x+1}{x+\gamma}\right)^d
\le\left(\frac{\beta x+1}{x}\right)^d
\le\left(\beta+\frac{\beta}{d(1-\beta\gamma)-3}\right)^d\\
&\le\beta^d\exp\left(\frac{d}{(1-\beta\gamma)d-3}\right)
<\frac{d(1-\beta\gamma)-3}{\beta},
\end{align*}
a contradiction.

\item[Case.2:] $x<\gamma$. Then $\beta x\le\beta\gamma<1$. Thus,
\begin{align*}
1\ge  \frac{d(1-\beta\gamma)}{\beta x+\frac{\gamma}{x}+2} \ge \frac{d(1-\beta\gamma)}{\frac{\gamma}{x}+3},
\end{align*}
which implies that $x\le \frac{\gamma}{d(1-\beta\gamma)-3}$. On the other hand,
\begin{align*}
x &=\left(\frac{\beta x+1}{x+\gamma}\right)^d
\ge\frac{1}{(x+1)^d}
\ge\left(1+\frac{\gamma}{d(1-\beta\gamma)-3}\right)^{-d}\\
&\ge\exp\left(-\frac{\gamma d}{d(1-\beta\gamma)-3}\right)
>\frac{\gamma}{d(1-\beta\gamma)-3},
\end{align*}
a contradiction.
\end{list}

We proceed to show that $\Gamma<\frac{1}{\beta}$. It is sufficient to show that there exists a $1<\gamma<\frac{1}{\beta}$ such that for all $d\ge 1$, $\frac{d(1-\beta\gamma)x}{(\beta x+1)(x+\gamma)}\le1$, where $x$ satisfies that $x=\left(\frac{\beta x+1}{x+\gamma}\right)^d$. 

If $\beta=0$, then $x=\left(\frac{1}{x+\gamma}\right)^d\le \frac{1}{\gamma^d}$. Thus,
\begin{align*}
\frac{d(1-\beta\gamma)x}{(\beta x+1)(x+\gamma)}= \frac{dx}{x+\gamma}\le dx \le \frac{d}{\gamma^d} \le\frac{1}{e\ln\gamma},
\end{align*}
where the last inequality can be verified by taking the maximum of $\frac{d}{\gamma^d}$ over $d$. Therefore,  setting $\gamma=e^{\frac{1}{e}}$, it holds that $\frac{d(1-\beta\gamma)x}{(\beta x+1)(x+\gamma)}\le \frac{1}{e\ln\gamma} = 1$.

On the other hand, if $0<\beta<1$,  choosing an arbitrary constant $\alpha\in(\exp(-\frac{1-\beta}{e}),1)$ which also satisfies that $\alpha\in(\beta,1)$,
and assuming $\gamma\in\left[\frac{1-(\alpha-\beta)}{\beta},\frac{1}{\beta}\right)\subseteq(1,\frac{1}{\beta})$, we have
\begin{align*}
x=\left(\frac{\beta x+1}{x+\gamma}\right)^d = \left(\beta+\frac{1-\beta\gamma}{x+\gamma}\right)^d\le \left(\beta+1-\beta\gamma\right)^d\le \alpha^d.
\end{align*}
Thus,
\begin{align*}
\frac{d(1-\beta\gamma)x}{(\beta x+1)(x+\gamma)}\le d(1-\beta\gamma)x \le  (1-\beta\gamma)d\alpha^d
\le\frac{(1-\beta\gamma)}{-e\ln\alpha},
\end{align*}
where the last inequality is also proved by taking the maximum of $d\alpha^d$. Therefore, we can choose $\gamma=\max\left\{\frac{1-(\alpha-\beta)}{\beta},\frac{1}{\beta}-\frac{e\ln(1/\alpha)}{\beta}\right\}$, which indeed satisifes $\gamma\in(1,\frac{1}{\beta})$, to guarantee that $\frac{d(1-\beta\gamma)x}{(\beta x+1)(x+\gamma)}\le \frac{(1-\beta\gamma)}{-e\ln\alpha} \le1$.

Therefore, for $0\le \beta<1$, there always exists a $1<\gamma<\frac{1}{\beta}$ such that for all $d\ge 1$, it holds that $\frac{d(1-\beta\gamma)x}{(\beta x+1)(x+\gamma)}\le1$, where $x$ satisfies that $x=\left(\frac{\beta x+1}{x+\gamma}\right)^d$. This implies $\Gamma<\frac{1}{\beta}$.
\end{proof}


\begin{definitions}
Let $\gamma(d)$ be the solution $\gamma$ of
\begin{align}
\frac{d(1-\beta\gamma)x(\gamma,d)}{(\beta x(\gamma,d)+1)(x(\gamma,d)+\gamma)} &=1 \label{eq:fix-point-2}
\end{align}
over $\gamma\in(1,\frac{1}{\beta})$, and define $\gamma(d)=1$ by convention if such solution does not exist.
\end{definitions}
The following lemma states that $\gamma(d)$ is well-defined and captures the uniqueness threshold for different instances of $d$.

\begin{lemma}\label{lemma-Gamma-gamma}
The followings hold for $\gamma(d)$:
\begin{enumerate}
\item
$\gamma(d)$ is a well-defined function for $d\ge 1$.
\item
$\Gamma=\sup_{d\ge1}\gamma(d)$.
\item
There exists a finite constant $D> 1$ such that $\Gamma=\gamma(D)$, and $D$ is a stationary point of $\gamma(d)$, i.e.~$\left.\deriv{\gamma}{d}\right|_{d=D}=0$.
\end{enumerate}
\end{lemma}
\begin{proof}
\begin{enumerate}
\item
We first show that for any $d\ge1$, there exists at most one $\gamma\in(1,\frac{1}{\beta})$ satisfying \eqref{eq:fix-point-2}, which will imply that $\gamma(d)$ is well-defined.

Observe that for any fixed $d\ge 1$, $x(\gamma,d)$ is strictly decreasing with respect to $\gamma$ over $\gamma\in(1,\frac{1}{\beta})$. By contradiction, assume that for some $d\ge 1$, $x$ is non-decreasing over $\gamma$. Then $x=\left(\frac{\beta x+1}{x+\gamma}\right)^d=\left(\beta+\frac{1-\beta\gamma}{x+\gamma}\right)^d$ is strictly decreasing over $\gamma$, a contradiction.

Therefore, $\frac{1-\beta\gamma}{x(\gamma,d)+\gamma}$ must be strictly decreasing with respect to $\gamma$, or otherwise $x=\left(\beta+\frac{1-\beta\gamma}{x+\gamma}\right)^d$ would have been non-decreasing, contradicting that $x(\gamma,d)$ is strictly decreasing.

Combining these together, we have
\begin{align*}
\frac{d(1-\beta\gamma)x(\gamma,d)}{(\beta x(\gamma,d)+1)(x(\gamma,d)+\gamma)}
=\frac{d(1-\beta\gamma)}{x(\gamma,d)+\gamma} \cdot \frac{1}{\beta+\frac{1}{x(\gamma,d)}}
\end{align*}
is strictly decreasing over $\gamma\in(1,\frac{1}{\beta})$. Thus, there exists at most one $\gamma\in(1,\frac{1}{\beta})$ satisfying \eqref{eq:fix-point-2}. Therefore, $\gamma(d)$ is well-defined.

\item
We then show that $\Gamma=\sup_{d\ge1}\gamma(d)$.
For any $d\ge 1$, let
\begin{align*}
\Gamma_d(\beta)=\inf\left\{\gamma\ge1 \Bigm{|} \frac{d(1-\beta\gamma)x(\gamma,d)}{(\beta x(\gamma,d)+1)(x(\gamma,d)+\gamma)}\le1\right\}.
\end{align*}
Note that for any $d\ge 1$, when $\gamma\rightarrow\frac{1}{\beta}$, $\frac{d(1-\beta\gamma)x(\gamma,d)}{(\beta x(\gamma,d)+1)(x(\gamma,d)+\gamma)}\rightarrow0<1$, thus $\Gamma_d(\beta)<\frac{1}{\beta}$. In addition to that, since $\frac{d(1-\beta\gamma)x(\gamma,d)}{(\beta x(\gamma,d)+1)(x(\gamma,d)+\gamma)}$ is strictly decreasing over $\gamma\in(1,\frac{1}{\beta})$, $\Gamma_d(\beta)$ is either equal to the unique solution $\gamma$ of $\frac{d(1-\beta\gamma)x(\gamma,d)}{(\beta x(\gamma,d)+1)(x(\gamma,d)+\gamma)}=1$ over $\gamma\in(0,\frac{1}{\beta})$ or equal to 1 if such solution does not exist. Therefore,
\begin{align*}
\Gamma_d(\beta)=\gamma(d).
\end{align*}

Since $\frac{d(1-\beta\gamma)x(\gamma,d)}{(\beta x(\gamma,d)+1)(x(\gamma,d)+\gamma)}$ is strictly decreasing over $\gamma\in(1,\frac{1}{\beta})$, for any $\gamma\in(1,\frac{1}{\beta})$ that $\gamma\ge\Gamma_d(\beta)$ for all $d\ge 1$, it holds that $\frac{d(1-\beta\gamma)x(\gamma,d)}{(\beta x(\gamma,d)+1)(x(\gamma,d)+\gamma)}\le 1$ for all $d\ge1$, i.e.~$\gamma\ge\Gamma(\beta)$.
Thus, $\Gamma(\beta)\le\sup_{d\ge 1}\Gamma_d(\beta)$. The other direction $\Gamma(\beta)\ge\sup_{d\ge 1}\Gamma_d(\beta)$ is universal. Therefore,
\begin{align*}
\Gamma(\beta)=\sup_{d\ge 1}\Gamma_d(\beta)=\sup_{d\ge1}\gamma(d).
\end{align*}

\item
We show that there is a finite $D>1$ that $\gamma(D)=\sup_{d\ge1}\gamma(d)$. 

First notice that $D>1$. By contradiction assume that $D=1$. Substituting $x$ in $(1-\beta\gamma)x=(\beta x+1)(x+\gamma)$ with the positive solution of $x=\left(\frac{\beta x+1}{x+\gamma}\right)$ gives us a $\gamma<1$. Then by conventional definition, $\gamma(D)=1$. From the previous analysis, we know that $\gamma(D)=\sup_{d\ge1}\gamma(d)=\Gamma$ and due to Lemma \ref{lemma-well-define-gamma}, $\Gamma>1$. A contradiction.


We treat $x=x(\gamma(d),d)$ as a single-variate function of $d$. We claim that $x\rightarrow0$ as $d\rightarrow\infty$. By contradiction, if $x$ is bounded away from 0 by a constant as $d\rightarrow\infty$, then $x=\left(\beta+\frac{1-\beta\gamma}{x+\gamma}\right)^d\le\left(\beta+\frac{1-\beta}{x+1}\right)^d\rightarrow 0$ as $d\rightarrow\infty$, a contradiction.

Therefore, when $d\rightarrow\infty$, it must hold that $\gamma(d)>1$, because if otherwise $\gamma\le 1$, since $x\rightarrow 0$ as $d\rightarrow\infty$, it holds that $x=\left(\frac{\beta x+1}{x+\gamma}\right)^d\rightarrow\frac{1}{\gamma^d}$, which approaches either $1$ or $\infty$ as $d\rightarrow\infty$, which contradicts that $x\rightarrow 0$ as $d\rightarrow\infty$.

We just show that $\gamma(d)>1$ for sufficiently large $d$, which means that for these $d$s, $\gamma(d)$ is  defined by \eqref{eq:fix-point-2} instead of defined by the convention $\gamma(d)=1$. Thus, for sufficiently large $d$, $x=x(\gamma(d),d)$ and $\gamma=\gamma(d)$ can be treated as single-variate functions of $d$ satisfying both \eqref{eq:fix-point-1} and \eqref{eq:fix-point-2}.

For $\beta\gamma<1$, it holds that $x=\left(\frac{\beta x+1}{x+\gamma}\right)^d\le\frac{1}{\gamma^d}$, thus
\begin{align*}
\frac{d}{\gamma^d}
\ge
\frac{d(1-\beta\gamma)}{\gamma^d}
\ge
d(1-\beta\gamma)x
=
(\beta x+1)(x+\gamma)
\ge 
\gamma,
\end{align*}
where the equality holds by \eqref{eq:fix-point-2}. Thus, $\gamma(d)\le d^{\frac{1}{d+1}}$.

Recall that $\gamma(d)>1$ for all sufficiently large $d$, thus there is a finite $d$ such that $\gamma(d)$ is bounded from below by a constant greater than 1. On the other hand, $\gamma(d)\le d^{\frac{1}{d+1}}=1$ as $d\rightarrow\infty$. Therefore, there is a finite $D$ such that $\gamma(D)=\sup_{d\ge 1}\gamma(d)$.
Due to Lemma \ref{lemma-Gamma-gamma}, this implies $\gamma(D)=\Gamma$.

Since $\gamma(D)=\sup_{d\ge 1}\gamma(d)$ and $D$ is neither infinite nor equal to 1, $D$ must be a stationary point of $\gamma(d)$, i.e.~$\left.\deriv{\gamma}{d}\right|_{d=D}=0$.
\end{enumerate}
\end{proof}

We can then define the crucial $D$ which generates the highest uniqueness threshold $\Gamma$.

\begin{definition}\label{definition-D-X}
Let $D$ be the value satisfying $\gamma(D)=\Gamma$. Let $X=x(\Gamma,D)$.
\end{definition}
It is obvious that both \eqref{eq:fix-point-1} and \eqref{eq:fix-point-2} hold for $\gamma=\Gamma$, $d=D$, and $x=X$. Two less obvious but very useful identities are given in the following lemma.

\begin{lemma}\label{lemma-fix-point-id}
The followings hold for $\Gamma,D$ and $X$.
\begin{align*}
1.\quad
&\frac{\beta}{\Gamma}\le\sqrt{\beta\Gamma}
<\frac{D-1}{D+1};\\
2.\quad
&\ln\left(\frac{\beta X+1}{X+\Gamma}\right)
=\frac{2(\beta X+1)}{(D+1)(\beta X+1)-2D}=\frac{2D(1-\beta\Gamma) X}{(\beta X+1)(2D X-(D+1)( X+\Gamma))}.
\end{align*}
\end{lemma}
\begin{proof}
\begin{enumerate}
\item
Since $\gamma=\Gamma$, $d=D$, and $x= X$ satisfies \eqref{eq:fix-point-2}, it holds that
\begin{align*}
D(1-\beta\Gamma)
&=\frac{(\beta X+1)( X+\Gamma)}{ X}
=\left(\beta  X+\frac{\Gamma}{ X}\right)+\beta\Gamma+1
\ge2\sqrt{\beta\Gamma}+\beta\Gamma+1
=\left(1+\sqrt{\beta\Gamma}\right)^2,
\end{align*}
where the inequality is due to the inequality of arithmetic and geometric means. Thus, $D\ge\frac{1+\sqrt{\beta\Gamma}}{1-\sqrt{\beta\Gamma}}$. Therefore,
\begin{align*}
\frac{D-1}{D+1}
&=1-\frac{2}{D+1}
\ge\sqrt{\beta\Gamma}
\ge\frac{\beta}{\Gamma},
\end{align*}
where the last inequality is implied trivially by that $0\le \beta<1$ and $\Gamma>1$.
\item
Recall that $ X=x(\Gamma, D)$ and $\Gamma=\gamma(D)$, where $x(\gamma,d)$ is defined by \eqref{eq:fix-point-1}, and $\gamma(d)$ is defined by \eqref{eq:fix-point-2}. Thus, $x=x(\gamma(d),d)$ and $\gamma=\gamma(d)$ can be treated as single-variate functions of $d$ satisfying both \eqref{eq:fix-point-1} and \eqref{eq:fix-point-2}.

The following identity is implied by \eqref{eq:fix-point-2}:
\begin{align}
d(1-\beta\gamma)x &=(\beta x+1)(x+\gamma). \label{eq:fix-point-4}
\end{align}
Taking the derivatives with respect to $d$ at $d=D$ for both sides of \eqref{eq:fix-point-4}, we have
\begin{align*}
\left.\left((\beta (x+\gamma)+(\beta x+1)-d(1-\beta\gamma))\deriv{x}{d}
+(\beta x(d+1)+1)\deriv{\gamma}{d}\right)\right|_{d=D}
=\left.(1-\beta\gamma)x\right|_{d=D}.
\end{align*}
Due to Lemma \ref{lemma-Gamma-gamma}, it holds that $\left.\deriv{\gamma}{d}\right|_{d=D}=0$. Then
\begin{align}
\left.\deriv{x}{d}\right|_{d=D}
&=
\frac{(1-\beta\Gamma) X}{\beta ( X+\Gamma)+(\beta  X+1)-D(1-\beta\Gamma)}\nonumber\\
&=
\frac{(1-\beta\Gamma) X^2}{\beta ( X+\Gamma) X+(\beta  X+1) X-(\beta X+1)( X+\Gamma)}
&\text{(applying \eqref{eq:fix-point-4})}\nonumber\\
&=
\frac{(1-\beta\Gamma) X^2}{\beta X^2-\Gamma}.
\label{eq:dx-dd}
\end{align}

Recall that $x(\gamma,d)$ is defined by \eqref{eq:fix-point-1}. Applying logarithm to both side of \eqref{eq:fix-point-1}, we have
\begin{align*}
\ln x = d\ln\left(\frac{\beta x+1}{x+\gamma}\right).
\end{align*}
Taking the partial derivatives with respect to $d$ for both sides,
\begin{align*}
\frac{1}{x}\pderiv{x}{d} &=  \ln\left(\frac{\beta x+1}{x+\gamma}\right)+\frac{\beta d}{(\beta x+1)}\cdot\pderiv{x}{d}-\frac{d}{(x+\gamma)}\cdot\pderiv{x}{d}.
\end{align*}
which implies that
\begin{align*}
\pderiv{x}{d}
&=\frac{x(\beta x+1)(x+\gamma)\ln\left(\frac{\beta x+1}{x+\gamma}\right)}{(\beta x+1)(x+\gamma)+d(1-\beta\gamma)x}\\
&=\frac{x(\beta x+1)(x+\gamma)\ln\left(\frac{\beta x+1}{x+\gamma}\right)}{(\beta x+1)(x+\gamma)+(\beta x+1)(x+\gamma)} & \text{(applying  \eqref{eq:fix-point-4})}\\
&=\frac{x}{2}\ln\left(\frac{\beta x+1}{x+\gamma}\right).
\end{align*}

Due to the total derivative formula, and that $\left.\deriv{\gamma}{d}\right|_{d=D}=0$,
\begin{align*}
\left.\deriv{x}{d}\right|_{d=D}
&=
\left.\pderiv{x(\gamma,d)}{\gamma}\cdot\deriv{\gamma(d)}{d}\right|_{d=D}+\left.\pderiv{x(\gamma,d)}{d}\right|_{d=D}
=0+\left.\pderiv{x(\gamma,d)}{d}\right|_{d=D}
=\frac{ X}{2}\ln\left(\frac{\beta  X+1}{ X+\Gamma}\right).
\end{align*}

Combining with \eqref{eq:dx-dd}, we have
\begin{align*}
\ln\left(\frac{\beta  X+1}{ X+\Gamma}\right)
&= \frac{2(1-\beta\Gamma) X}{\beta X^2-\Gamma}
\end{align*}
The equations in the lemma are consequences of the above equation. Specifically,
\begin{align*}
\frac{2(1-\beta\Gamma) X}{\beta X^2-\Gamma}
&= \frac{2D(1-\beta\Gamma) X}{D(\beta X-1)( X+\Gamma)+D(1-\beta\Gamma) X}\\
&= \frac{2(\beta X+1)( X+\Gamma)}{D(\beta X-1)( X+\Gamma)+(\beta X+1)( X+\Gamma)} &\text{(applying \eqref{eq:fix-point-4})}\\
&=\frac{2(\beta X+1)}{(D+1)(\beta X+1)-2D};
\end{align*}
and
\begin{align*}
\frac{2(1-\beta\Gamma) X}{\beta X^2-\Gamma}
&= \frac{2D(1-\beta\Gamma) X}{D(\beta X+1)( X-\Gamma)-D(1-\beta\Gamma) X}\\
&= \frac{2D(1-\beta\Gamma) X}{D(\beta X+1)( X-\Gamma)-(\beta X+1)( X+\Gamma)} &\text{(applying \eqref{eq:fix-point-4})}\\
&=\frac{2D(1-\beta\Gamma) X}{(\beta X+1)(2D X-(D+1)( X+\Gamma))}.
\end{align*}
\end{enumerate}
\end{proof}

}

\section{Computationally Efficient Correlation Decay}\label{section-analysis}
We prove Theorem \ref{thm:rank-decay}, justifying the computationally efficient correlation decay.

We use $R_v$ and $R_v+\delta_v$ to respectively denote the lower and upper bounds of $R_T^{\sigma_{\Lambda}}$ where $T$ is rooted by $v$.  For fixed vertices $v\in B_M^*(L)$, set $R_v=\infty$ if $v$ is blue (and $R_v=0$ if $v$ is green) and $\delta=0$.
For all free vertices $v\in B_M^*(L)$, supposed that $v$ has $d_1$ children fixed to be blue, $d_0$ children fixed to be green, and $d$ free children $v_1,\ldots,v_d$, the recursion \eqref{eq:recursion} gives that
\begin{align}
R_v+\delta_v=  \frac{\beta^{d_1}}{\gamma^{d_0}} \prod_{i=1}^d \frac{\beta R_{v_i} +1 }{R_{v_i} +\gamma}\ \ \ \mbox{ and }  \ \ \ \ R_v=  \frac{\beta^{d_1}}{\gamma^{d_0}} \prod_{i=1}^d \frac{\beta (R_{v_i} +\delta_{v_i})+1 }{R_{v_i}+\delta_{v_i} +\gamma}.\label{eq:R-recursion}
\end{align}
And for all vertices $v\not\in B_M^*(L)$, we use the naive bounds that $R_v=0$ and $\delta_v=\infty$.

Since $\gamma>\Gamma>1>\beta\ge0$, the range of the recursion is $(0,1]$ as long as the inputs are positive. Thus for all free vertices $v\in B_M^*(L)$, it holds that $0<R_v\leq R_v+\delta_v \le1$.

Due to the monotonicity of the recursion, denoted by $r$ the root of the tree,  $R_r$ and $R_r+\delta_r$ are lower and upper bounds respectively for all $R_T^{\tau_\Lambda}$ where $\tau_\Lambda=\sigma_\Lambda$ in $B_M^*(L)$. Theorem \ref{thm:rank-decay} is then implied by that $\delta_r\leq \exp(-\Omega(L))$.

Let $f(x_1,\ldots,x_d)= \frac{\beta^{d_1}}{\gamma^{d_0}} \prod_{i=1}^d \frac{\beta x_i +1 }{x_i +\gamma}$. Then the recursions \eqref{eq:R-recursion} can be written as that $R_v+\delta_v=f(R_{v_1},\ldots,R_{v_d})$ and $R_v=f(R_{v_1}+\delta_{v_1},\ldots,R_{v_d}+\delta_{v_d})$. Due to the Mean Value Theorem, there exist $\widetilde{R_i}\in[R_{v_i}, R_{v_i}+\delta_{v_i}]$, $1\le i\le d$, such that
\begin{align*}
\delta_v
& =    
f(R_{v_1},\ldots,R_{v_d})-f(R_{v_1}+\delta_{v_1},\ldots,R_{v_d}+\delta_{v_d})\\
&=
-\nabla f(\widetilde{R_1},\ldots, \widetilde{R_d})\cdot(\delta_{v_1},\ldots,\delta_{v_d})\\
&=
\frac{\beta^{d_1}}{\gamma^{d_0}} (1-\beta \gamma) \cdot \prod_{i=1}^d \frac{\beta \widetilde{R_i}+1} {\widetilde{R_i}+\gamma}  \cdot \sum_{i=1}^d  \frac{\delta_{v_i}}{  (\beta \widetilde{R_i} +1)(\widetilde{R_i}+\gamma)}.
\end{align*}


A straightforward estimation gives that
\begin{align*}
\frac{\delta_v}{\max_{1\le i\le d}\{\delta_{v_i}\}}
&\leq
\frac{\beta^{d_1}}{\gamma^{d_0}} (1-\beta \gamma) \cdot \prod_{i=1}^d \frac{\beta \widetilde{R_i}  +1} {\widetilde{R_i}+\gamma}  \cdot \sum_{i=1}^d  \frac{1}{  (\beta \widetilde{R_i} +1)(\widetilde{R_i}+\gamma)}.
\end{align*}
If this ratio is  bounded by a constant less than 1, then the gap $\delta$ shrinks by a constant factor for each step of recursion, thus an exponential decay would have been established. However, such a step-wise guarantee of decay holds in general only when the $\gamma$  is substantially greater than $\Gamma(\beta)$. A simulation shows that when $\gamma$ is sufficiently close to $\Gamma(\beta)$, the gap $\delta$ may indeed increase for some specific $d$ and $R_i$. We then apply an amortized analysis to show that even though the gap $\delta$ may occasionally increase, it decays exponentially in a long run.

\subsection{Amortized analysis of correlation decay}
We use the potential method to analyze the amortized behavior of correlation decay. The potential function is defined as
\begin{align*}
\Phi(R)
&=
R^{\frac{D+1}{2D}}(\beta R+1),
\end{align*}
\ifabs{where $D$ is the crucial $d$ which generates the highest uniqueness threshold. Formally, the infimum of such $\gamma$ that satisfies  $\frac{d(1-\beta\gamma)(\beta \hat{x}+1)^{d-1}}{(\hat{x}+\gamma)^{d+1}}\le 1$ where $\hat{x}=\left(\frac{\beta\hat{x} + 1}{\hat{x} + \gamma}\right)^d$ is maximized and equal to $\Gamma(\beta)$ when $d=D$.}{where $D$ is the crucial $d$ which generates the highest uniqueness threshold as formally defined in Section \ref{app-Gamma-D}}


We will analyze the decay rate of  $\frac{\delta}{\Phi}$ instead of $\delta$. This is done by introducing a monotone function $\varphi(R)$, which is implicitly defined by its derivative $\varphi'(R)=\frac{1}{\Phi(R)}$. We denote that $y_v=\varphi(R_v)$ and $y_v+\epsilon_v=\varphi(R_v+\delta_v)$.
Recall that $R_v+\delta_v=f(R_{v_1},\ldots,R_{v_d})$ and $R_v=f(R_{v_1}+\delta_{v_1},\ldots,R_{v_d}+\delta_{v_d})$ where $f(x_1,\ldots,x_d)= \frac{\beta^{d_1}}{\gamma^{d_0}} \prod_{i=1}^d \frac{\beta x_i +1 }{x_i +\gamma}$. Then
\begin{align*}
y_v &=\varphi(R_v)\\
&=\varphi\left(f(R_{v_1}+\delta_{v_1},\ldots,R_{v_d}+\delta_{v_d})\right)\\
&=\varphi\left(f(\varphi^{-1}(y_{v_1}+\epsilon_{v_1}),\ldots,\varphi^{-1}(y_{v_d}+\epsilon_{v_d}))\right); \quad\text{and}\\
y_v+\epsilon_v &=\varphi(R_v+\delta_v)\\
&=\varphi\left(f(R_{v_1},\ldots,R_{v_d})\right)\\
&=\varphi\left(f(\varphi^{-1}(y_{v_1}),\ldots,\varphi^{-1}(y_{v_d}))\right).
\end{align*}
By the Mean Value Theorem, there exists an $\widetilde{R}\in[R_v,R_v+\delta_v]$ such that
\begin{align}
\epsilon_v = \varphi(R_v+\delta_v)-\varphi(R_v) = \delta_v\cdot \varphi'(\widetilde{R})=\frac{\delta_v}{\Phi(\widetilde{R})}.\label{eq:epsilon-delta}
\end{align}
By the Mean Value Theorem, there exist $\widetilde{R_i}\in[R_{v_i},R_{v_i}+\delta_{v_i}]$, $1\le i\le d$, such that
\begin{align}
\epsilon_v
&=
\varphi\left(f(R_{v_1},\ldots,R_{v_d})\right)-\varphi\left(f(R_{v_1}+\delta_{v_1},\ldots,R_{v_d}+\delta_{v_d})\right) \notag\\
&=
-\nabla \varphi\left(f(\widetilde{R_1},\ldots,\widetilde{R_d})\right)\cdot(\delta_{v_1},\ldots,\delta_{v_d}) \notag\\
&=
(1-\beta \gamma) \cdot \frac{\left( \frac{\beta^{d_1}}{\gamma^{d_0}} \prod_{i=1}^d \frac{\beta \widetilde{R_i}  +1} {\widetilde{R_i}+\gamma} \right)^{\frac{D-1}{2D}} }{\beta  \frac{\beta^{d_1}}{\gamma^{d_0}} \prod_{i=1}^d \frac{\beta \widetilde{R_i}  +1} {\widetilde{R_i}+\gamma} +1}
  \cdot \sum_{i=1}^d  \frac{\delta_{v_i}}{  (\beta \widetilde{R_i} +1)(\widetilde{R_i} +\gamma)} \notag\\
&\le
\frac{d}{\gamma^{(d_0+d_1+d)\frac{D-1}{2D}}},  \label{eq:epsilon-bootstrap}
\end{align}
where \eqref{eq:epsilon-bootstrap} is trivially implied by that $\widetilde{R_i}\in(0, 1]$, $\gamma>1$ and $\beta\gamma<1$.
By the Mean Value Theorem, there exist $\widetilde{y_i}\in[y_{v_i},y_{v_i}+\epsilon_{v_i}]$ and due to the monotonicity of $\varphi(\cdot)$, corresponding  $\widetilde{R_i}\in[R_{v_i},R_{v_i}+\delta_{v_i}]$ that $\widetilde{y_i}=\varphi(\widetilde{R_i})$, $1\le i\le d$, such that
\begin{align}
\epsilon_v
&=
\varphi\left(f(\varphi^{-1}(y_1),\ldots,\varphi^{-1}(y_d))\right)-\varphi\left(f(\varphi^{-1}(y_1+\epsilon_1),\ldots,\varphi^{-1}(y_d+\epsilon_d))\right) \notag\\
&=
-\nabla \varphi\left(f(\varphi^{-1}(\widetilde{y_1}),\ldots,\varphi^{-1}(\widetilde{y_d}))\right)\cdot(\epsilon_1,\ldots,\epsilon_d) \notag\\
&=
\frac{(1-\beta \gamma) \left( \frac{\beta^{d_1}}{\gamma^{d_0}} \prod_{i=1}^d \frac{\beta \widetilde{R_i} +1} {\widetilde{R_i} +\gamma} \right)^{\frac{D-1}{2D}} }{\beta  \frac{\beta^{d_1}}{\gamma^{d_0}} \prod_{i=1}^d \frac{\beta \widetilde{R_i} +1} {\widetilde{R_i} +\gamma} +1}\cdot \sum_{i=1}^d\frac{\widetilde{R_i}^{\frac{D+1}{2D}} \cdot \epsilon_{v_i}}{\widetilde{R_i} +\gamma } \notag\\
&\le
\max_{1\le i\le d}\{\epsilon_{v_i}\}\cdot \frac{(1-\beta \gamma) \left( \frac{\beta^{d_1}}{\gamma^{d_0}} \prod_{i=1}^d \frac{\beta \widetilde{R_i} +1} {\widetilde{R_i} +\gamma} \right)^{\frac{D-1}{2D}} }{\beta  \frac{\beta^{d_1}}{\gamma^{d_0}} \prod_{i=1}^d \frac{\beta \widetilde{R_i} +1} {\widetilde{R_i} +\gamma} +1}\cdot \sum_{i=1}^d  \frac{\widetilde{R_i}^{\frac{D+1}{2D}} }{\widetilde{R_i} +\gamma }. \label{eq:bootstrap}
\end{align}
Since $\widetilde{R_i}\in(0,1]$, $\gamma>1$, and $\beta\gamma<1$, \eqref{eq:bootstrap} trivially implies that
\begin{align}
\epsilon_v
&\le 
\max_{1\le i\le d}\{\epsilon_{v_i}\}\cdot d\left(\frac{\beta^{d_0}}{\gamma^{d_1}}\prod_{i=1}^d \frac{\beta \widetilde{R_i} +1} {\widetilde{R_i} +\gamma}\right)^{\frac{D-1}{2D}}\le \frac{d}{\gamma^{(d_0+d_1+d)\frac{D-1}{2D}}}\cdot  \max_{1\le i\le d}\{\epsilon_{v_i}\},
\label{eq:step-1}
\end{align}
On the other hand, we know that $\beta\le\sqrt{\beta\Gamma}<\frac{D-1}{D+1}$ (\ifabs{the proof is in full version}{due to Lemma \ref{lemma-fix-point-id} in Section \ref{app-Gamma-D}}).  It is easy to verify that function $\frac{x^\frac{D-1}{2D}}{\beta x+1}$ is monotonically increasing when $x\le1$. Then the following is also implied by \eqref{eq:bootstrap}:
\begin{align}
\epsilon_v
&\le
\alpha(d; \widetilde{R_1},\ldots,\widetilde{R_d})\cdot\max_{1\le i\le d}\{\epsilon_{v_i}\}.\label{eq:amortized-decay}
\end{align}
where the function $\alpha(d;x_1,\ldots,x_d)$ captures the amortized decay, defined as
\begin{align}
\alpha(d;x_1,\ldots,x_d)
&=\frac{(1-\beta \gamma) \left( \prod_{i=1}^d \frac{\beta x_i +1} {x_i +\gamma} \right)^{\frac{D-1}{2D}} }{\beta \prod_{i=1}^d \frac{\beta x_i+1} {x_i+\gamma} +1}\cdot \sum_{i=1}^d  \frac{x_i^{\frac{D+1}{2D}} }{x_i +\gamma }. \label{eq:alpha}
\end{align}

\ifabs{}{
Our goal is to upper bound the $\alpha(d;x_1,\ldots,x_d)$ assuming the uniqueness condition.
A concave analysis reduces the upper bound to the symmetric cases that all $x_i$ are equal.

\begin{lemma}\label{lemma-symmetrization}
Let  $0\le\beta<1$, $\gamma>\Gamma(\beta)$, and $\beta\gamma<1$. Then for any $d\ge 1$ and any $x_1,\ldots,x_d\in (0,1]$, there exists an $x\in(0,1]$, such that $\alpha(d;x_1,\ldots,x_d)$ is maximized when all $x_i=x$.
\end{lemma}
\begin{proof}
We denote $y_i= \ln(\frac{\beta x_i +1} {x_i +\gamma})$, then $x_i= \frac{1-\beta \gamma}{e^{y_i}-\beta}-\gamma$ and
\begin{align*}
\alpha(d;x_1,\ldots,x_d)
&=
\frac{(1-\beta \gamma)\exp\left(\frac{D-1}{2D}\sum_{i=1}^d y_i \right) }{\beta  \exp\left(\sum_{i=1}^d y_i\right) +1}\cdot \sum_{i=1}^d  \frac{\left(\frac{1-\beta \gamma}{e^{y_i}-\beta}-\gamma\right)^{\frac{D+1}{2D}} }{\frac{1-\beta \gamma}{e^{y_i}-\beta} }\\
&=\frac{\exp\left(\frac{D-1}{2D}\sum_{i=1}^d y_i \right)}{\beta  \exp\left(\sum_{i=1}^d y_i\right) +1}\cdot \sum_{i=1}^d f(y_i),
\end{align*}
where $f(y)= \left(\frac{1-\beta \gamma}{e^{y}-\beta}-\gamma\right)^{\frac{D+1}{2D}} (e^{y}-\beta)$.

It holds that
\begin{align*}
f'(y)
&= e^y \left(\frac{1-\beta \gamma}{e^{y}-\beta}-\gamma\right)^\frac{D+1}{2D} \left(1+\frac{\frac{D+1}{2D}(1-\beta\gamma)}{\gamma e^y-1}\right),\\
f''(y)
&=\frac{ e^y \left(\frac{1-\beta \gamma}{e^{y}-\beta}-\gamma\right)^\frac{D+1}{2D}}
  {4 D^2 (e^y-\beta) (\gamma e^y-1)^2 } \cdot g(y,D),
\end{align*}
  where
\begin{align*}
g(y,D)
&=e^y (\beta \gamma -1)^2 -2(1-\beta \gamma)(e^y-\beta)(1-\gamma e^y)D\\
&\quad\,- (2 \beta + 2 \beta^2 \gamma -e^y(1 +10 \beta \gamma +\beta^2 \gamma^2)+6\gamma e^{2y}(1+\beta \gamma)- 4 \gamma e^{3y})D^2.
\end{align*}
The fact $e^y\in (\beta, \frac{1}{\gamma})$ implies that the sign of $f''(y)$ is the same as that of $g(y,D)$. In the follow, we show that $g(y,D)$ is always negative.
The coefficient of $D$ in $g(y)$ is obviously negative
  given that $e^y\in (\beta, \frac{1}{\gamma})$. Now we show that the  coefficient of $D^2$ in $g(y)$ is also negative. To show this, the condition $e^y\in (\beta, \frac{1}{\gamma})$ is not sufficient.  We substitute $y_i= \ln(\frac{\beta x_i +1} {x_i +\gamma})$ back and recall that $x_i\in(0,1)$, we have
\begin{align*}
&\quad\,\, 2 \beta + 2 \beta^2 \gamma -e^y(1 +10 \beta \gamma +\beta^2 \gamma^2) +6\gamma e^{2y}(1+\beta \gamma)- 4 \gamma e^{3y}\\
&= \frac{(\beta\gamma -1)^2}{(\gamma +x)^3}\cdot(\gamma^2-x^2+ \gamma(1-\beta \gamma) x+3 \gamma x + 4 \beta \gamma x^2 + \beta x^3)
>0.
\end{align*}

Since both the coefficients of $D$ and $D^2$ are negative, we can choose $D=1$, in which case,
\[f''(y)=-\gamma e^y <0.\]
Denote that $\bar{y}=\frac{1}{d}\sum_{i=1}^d y_i$. Due to the Jensen's Inequality, $\sum_{i=1}^d f(y_i)\le df(\bar{y})$.
Therefore,
\begin{align*}
\alpha(d;x_1,\ldots,x_d)
&=\frac{\exp\left(\frac{D-1}{2D}\sum_{i=1}^d y_i \right)}{\beta  \exp\left(\sum_{i=1}^d y_i\right) +1}\cdot \sum_{i=1}^d f(y_i)\\
&\le \frac{\exp\left(\frac{d(D-1)}{2D}\bar{y} \right)}{\beta  \exp\left(d\bar{y}\right) +1}\cdot d f(\bar{y}).
\end{align*}
Let $x$ satisfy that $\bar{y}= \ln(\frac{\beta x +1} {x +\gamma})$, i.e.~$x= \frac{1-\beta \gamma}{e^{\bar{y}}-\beta}-\gamma$. It is then easy to verify that $x\in(0,1]$ since all $x_i\in(0,1]$ and $y_i= \ln(\frac{\beta x_i +1} {x_i +\gamma})$ is monotone with respect to $x_i$. Therefore, $\alpha(d;x_1,\ldots,x_d)$ is maximized when all $x_i=x\in(0,1]$.
\end{proof}

We then deal with the symmetric case.
Let
\begin{align*}
\alpha(d,x)
&=\alpha(d;\underbrace{x,\ldots,x}_{d})
=
\frac{d(1-\beta\gamma)x^{\frac{D+1}{2D}}(\beta x+1)^{\frac{d(D-1)}{2D}}}{
(x+\gamma)^{1+\frac{d(D-1)}{2D}}\left(\beta\left(\frac{\beta x+1}{x+\gamma}\right)^d+1\right)}.
\end{align*}
Let $f(x)=\left(\frac{\beta x+1}{x+\gamma}\right)^d$ be the symmetric version of the recursion \eqref{eq:recursion}. Observe that $\alpha(d,x)=\frac{\Phi(x)}{\Phi(f(x))}|f'(x)|$, which is exactly the amortized decay ratio in the symmetric case. 

Recall the formal definitions of $D$ and $X$ in Definition \ref{definition-D-X} in Section \ref{app-Gamma-D}. 
Our main discovery is the following lemma which states that at the uniqueness threshold $\gamma=\Gamma(\beta)$, the value of  $\alpha(d,x)$ is maximized at $d=D$ and $x=X$ with $\alpha(D,X)=1$. It is in debt to the magic of the potential method to observe such a harmoniously beautiful coincidence between amortized correlation decay and phase transition of uniqueness. 


\begin{lemma}\label{lemma-symmetric-extreme}
Let $0\le\beta<1$ and $\gamma=\Gamma(\beta)$. It holds that $\sup_{d\ge1\atop 0<x\le1}\alpha(d,x)=\alpha(D, X)=1$.
\end{lemma}
\begin{proof}
It is not difficult to verify that $\alpha(D, X)=1$. Note that \eqref{eq:fix-point-1} and \eqref{eq:fix-point-2} hold for $\gamma=\Gamma, d=D$ and $x= X$. Then
\begin{align*}
\alpha(D, X)
&=
\frac{D(1-\beta\Gamma)\left(\frac{\beta X+1}{ X+\Gamma}\right)^D}{(\beta X+1)( X+\Gamma)}\cdot\left(\frac{ X}{\left(\frac{\beta X+1}{ X+\Gamma}\right)^D}\right)^{\frac{D+1}{2D}}\cdot\left(\frac{\beta X+1}{\beta\left(\frac{\beta X+1}{ X+\Gamma}\right)^D+1}\right)\\
&=\frac{D(1-\beta\Gamma) X}{(\beta X+1)( X+\Gamma)}\left(\frac{ X}{ X}\right)^{\frac{D+1}{2D}}\left(\frac{\beta X+1}{\beta X+1}\right)\\
&=1.
\end{align*}

We then show that $\sup_{d\ge1\atop 0<x\le1}\alpha(d,x)=\alpha(D, X)$. For the rest of the proof, we assume that $0\le\beta<1$, $d\ge 1$, and $0<x\le1$. Due to lemma \ref{lemma-well-define-gamma}, $1<\Gamma(\beta)<\frac{1}{\beta}$. And we know that $ X=x(\Gamma,D)\in(0,1)$.

Denote that $z=\left(\frac{\beta x+1}{x+\Gamma}\right)^d$. Then $\alpha(d,x)$ can be rewritten as $\alpha(d,x)=C_1\cdot\frac{dz^{\frac{D-1}{2D}}}{\left(\beta z+1\right)}$,
where $C_1=\frac{1-\beta\Gamma}{x+\Gamma}x^{\frac{D+1}{2D}}>0$ is independent of $d$. Thus,
\begin{align*}
\frac{\partial \alpha(d,x)}{\partial d}
&=
\frac{C_1 z^{\frac{D-1}{2D}}}{2D\left(\beta z+1\right)^2}\cdot g(z),
\end{align*}
where the function $g(z)$ is defined as
\[
g(z)=2D\left(\beta z+1\right)-\left((D+1)(\beta z+1)-2D\right)\ln z.
\] 
It is obvious that $\frac{C_1z^{\frac{D-1}{2D}}}{2D\left(\beta z+1\right)^2}>0$, thus the sign of $\pderiv{\alpha(d,x)}{d}$ is governed by $g(z)$.
Note that $0<z=\left(\frac{\beta x+1}{x+\Gamma}\right)^d<1$. Then
\begin{align*}
\deriv{g(z)}{z}
&=\frac{1}{z}\left((D-1)(\beta z + 1) - (D+1)\beta z\ln z\right)
>(D-1)(\beta z+1)
\ge 0
\end{align*}
Therefore, $g(z)$ is strictly increasing with respect to $z$.
Due to Lemma \ref{lemma-fix-point-id}, it holds that $\ln\left(\frac{\beta  X+1}{ X+\Gamma}\right)=\frac{2(\beta X+1)}{(D+1)(\beta X+1)-2D}$. Thus,
\begin{align*}
g( X)
&= 2D\left(\beta X+1\right)-D\left((D+1)(\beta X+1)-2D\right)\ln\left(\frac{\beta  X+1}{ X+\Gamma}\right)
=0.
\end{align*}
Therefore, $\pderiv{\alpha(d,x)}{d}<0$ when $z< X$; $\pderiv{\alpha(d,x)}{d}=0$ when $z= X$; and $\pderiv{\alpha(d,x)}{d}>0$ when $z> X$. Note that $z=\left(\frac{\beta x+1}{x+\Gamma}\right)^d$ is monotonically decreasing with respect to $d$ since $\left(\frac{\beta x+1}{x+\Gamma}\right)<\frac{1}{\Gamma}<1$. Let $\rho(x)=\frac{\ln X}{\ln(\beta x+1)-\ln(x+\Gamma)}$. It is then easy to verify that
\begin{align*}
\pderiv{\alpha(d,x)}{d}
&\begin{cases}
<0 & \text{if }d>\rho(x),\\
=0 & \text{if }d=\rho(x),\\
>0 & \text{if }d<\rho(x).
\end{cases}
\end{align*}
Therefore, for any $d$ and $x$, $\alpha(d,x)\le\alpha(\rho(x),x)$.

Recall that $\alpha(d,x)=\frac{(1-\beta\Gamma)x^{\frac{D+1}{2D}}}{x+\Gamma}\cdot\frac{dz^{\frac{D-1}{2D}}}{\left(\beta z+1\right)}$,
where $z=\left(\frac{\beta x+1}{x+\Gamma}\right)^d$.
When $d=\rho(x)=\frac{\ln X}{\ln(\beta x+1)-\ln(x+\Gamma)}$, it holds that $z=\left(\frac{\beta x+1}{x+\Gamma}\right)^d= X$. Therefore,
\begin{align*}
\alpha(\rho(x),x)
&= C_2\cdot\frac{x^{\frac{D+1}{2D}}}{(x+\Gamma)(\ln(\beta x+1)-\ln(x+\Gamma))},
\end{align*}
where $C_2=\frac{(1-\beta\Gamma) X^{\frac{D-1}{2D}}\ln X}{\left(\beta X+1\right)}$ is independent of $x$, and $C_2<0$ since $0< X<1$.
\begin{align*}
\deriv{\alpha(\rho(x),x)}{x}
&= \frac{C_2x^{\frac{-D+1}{2D}}\cdot h(x)}{2D(x+\Gamma)^2(\beta x+1)\left(\ln\left(\frac{\beta x+1}{x+\Gamma}\right)\right)^2},
\end{align*}
where 
\begin{align*}
h(x)
&=
2D(1-\beta\Gamma)x-(\beta x+1)(2D x-(D+1)(x+\Gamma))\cdot \ln\left(\frac{\beta x+1}{x+\Gamma}\right).
\end{align*}
It is easy to see that $ \frac{C_2x^{\frac{-D+1}{2D}}}{2D(x+\Gamma)^2(\beta x+1)\left(\ln\left(\frac{\beta x+1}{x+\Gamma}\right)\right)^2}<0$ and $h(x)$ is monotonically increasing. Due to Lemma \ref{lemma-fix-point-id}, $\ln\left(\frac{\beta  X+1}{ X+\Gamma}\right)=\frac{2D(1-\beta\Gamma) X}{(\beta X+1)(2D X-(D+1)( X+\Gamma))}$, thus $h( X)=0$.
Therefore, $\deriv{\alpha(\rho(x),x)}{x}$ is monotonically decreasing with respect to $x$ and $\left.\deriv{\alpha(\rho(x),x)}{x}\right|_{x= X}=0$, which implies that for any $x$, $\alpha(\rho(x),x)\le\alpha(\rho( X), X)$.

Due to \eqref{eq:fix-point-1}, it holds that $ X=\left(\frac{\beta X+1}{ X+\Gamma}\right)^D$, thus
$\rho( X)=
\frac{\ln X}{\ln\left(\frac{\beta X+1}{ X+\Gamma}\right)}
=D$, hence $\alpha(\rho( X), X)=\alpha(D, X)$.

In conclusion, assuming $0\le\beta<1$ and $\gamma=\Gamma(\beta)$, for any $d\ge1$ and $0<x\le 1$, it holds that
\begin{align*}
\alpha(d,x)
&\le\alpha(\rho(x),x)
\le\alpha(\rho( X), X)
=\alpha(D, X)
=1.
\end{align*}
\end{proof}

As a consequence of the above lemma, a strict upper bound is obtained as follows.

\begin{lemma}\label{lemma-symmetric-decay}
For $0\le\beta<1$ and $\Gamma(\beta)<\gamma<\frac{1}{\beta}$,  there exists a constant $\alpha<1$ such that for any $d\ge1$ and $0<x\le1$, it holds that $\alpha(d,x)\le \alpha$.
\end{lemma}
\begin{proof}
Let $\alpha_{\beta,\gamma}=\sup_{d\ge 1\atop0<x\le1}\alpha(d,x)$. Note that $\alpha_{\beta,\gamma}$ is a constant independent of $d$ and $x$. And $\alpha_{\beta,\Gamma}=1$ due to Lemma \ref{lemma-symmetric-extreme}.

We then show that  $\alpha_{\beta,\gamma}<\alpha_{\beta,\Gamma}$ for $\Gamma<\gamma<\frac{1}{\beta}$.
In particular, we first show that for any $d\ge1$ and $0<x\le 1$, $\alpha(d,x)$ is strictly decreasing with respect to $\gamma$ over $\gamma\in(\Gamma,\frac{1}{\beta})$.
\begin{align*}
\alpha(d,x)
&=C_3\cdot\left(\frac{1-\beta\gamma}{x+\gamma}\right)\cdot\frac{(x+\gamma)^\frac{d(D+1)}{2D}}{\left(\beta(\beta x+1)^d+(x+\gamma)^d\right)},
\end{align*}
where $C_3=dx^{\frac{d(D+1)}{2D}}(\beta x+1)^{\frac{d(D-1)}{2D}}>0$ is independent of $\gamma$. Let $h(\gamma)=\frac{(x+\gamma)^\frac{d(D+1)}{2D}}{\left(\beta(\beta x+1)^d+(x+\gamma)^d\right)}$.
\begin{align*}
\deriv{h(\gamma)}{\gamma}
&=
\frac{d(D+1)(x+\gamma)^{\frac{d(3D+1)}{2D}-1}}{2D\left(\beta(\beta x+1)^d+(x+\gamma)^d\right)^2}
\cdot\left(\beta\left(\frac{\beta x+1}{x+\gamma}\right)-\frac{D-1}{D+1}\right)\\
&<\frac{d(D+1)(x+\gamma)^{\frac{d(3D+1)}{2D}-1}}{2D\left(\beta(\beta x+1)^d+(x+\gamma)^d\right)^2}
\left(\frac{\beta}{\Gamma}-\frac{D-1}{D+1}\right)\\
&<0,
\end{align*}
where the second to the last inequality holds because $x>0$ and $\gamma>\Gamma$, and the last inequality is due to Lemma \ref{lemma-fix-point-id}. The fact  that $\deriv{h(\gamma)}{\gamma}<0$ implies that $h(\gamma)$ is strictly decreasing.
Thus, $\alpha(d,x)=C_3\cdot\left(\frac{1-\beta\gamma}{x+\gamma}\right)\cdot h(\gamma)$ is strictly decreasing with respect to $\gamma$ over $\gamma\in(\Gamma,\frac{1}{\beta})$.

Let $\alpha_\gamma(d,x)$ denote the $\alpha(d,x)$ with parameter $\gamma$.
We can assume that there exist finite $d^*\ge 1$ and  constant $0<x^*\le1$ achieving that $\alpha_\gamma(d^*,x^*)=\sup_{d\ge1\atop0<x\le1}\alpha_\gamma(d,x)$, since if otherwise it would hold that $\sup_{d\ge1\atop0<x\le1}\alpha_\gamma(d,x)$ is achieved by either $d\rightarrow\infty$ or $x\rightarrow0$, but in either case it is easy to verify that $\alpha_\gamma(d,x)\rightarrow 0$, thus $\alpha_{\beta,\gamma}=\sup_{d\ge1\atop0<x\le1}\alpha_\gamma(d,x)=0$ and we are done. Since $\alpha(d,x)$ is strict decreasing with respect to $\gamma$, it holds that $\alpha_\gamma(d^*,x^*)<\alpha_\Gamma(d^*,x^*)$ for any $\gamma\in(\Gamma,\frac{1}{\beta})$  Therefore,
\begin{align*}
\alpha_{\beta,\gamma}
&=
\sup_{d\ge1\atop0<x\le1}\alpha_\gamma(d,x)
=
\alpha_\gamma(d^*,x^*)
<
\alpha_\Gamma(d^*,x^*)
\le
\sup_{d\ge1\atop0<x\le1}\alpha_\Gamma(d,x)
=\alpha_{\beta,\Gamma}
=1.
\end{align*}
\end{proof}

Combining Lemma \ref{lemma-symmetrization} and Lemma \ref{lemma-symmetric-decay}, we have the following lemma which bounds the amortized correlation decay when the uniqueness is satisfied.
}
\begin{lemma}\label{lemma-alpha}
Let $\alpha(d;x_1,\ldots,x_d)$ be defined by \eqref{eq:alpha}.
For $0\le\beta<1$ and $\Gamma(\beta)<\gamma<\frac{1}{\beta}$,  there exists a constant $\alpha<1$ which depends only on $\beta$ and $\gamma$, such that for any $d\ge1$ and $x_i\in(0,1]$, $i=1,2,\ldots,d$, it holds that $\alpha(d;x_1,\ldots,x_d)\le \alpha$
\end{lemma}
\begin{proofsketch}
By Jensen's inequality, $\alpha(d;x_1,\ldots,x_d)$ is upper bounded by its symmetric form, when all $x_i=x$. Denote this symmetric form as $\alpha(d,x)=\alpha(d;x,\ldots,x)$, and the symmetric version of the recursion \eqref{eq:recursion} as $f(x)=\left(\frac{\beta x+1}{x+\gamma}\right)^d$. Observe that $\alpha(d,x)=\frac{\Phi(x)}{\Phi(f(x))}|f'(x)|$, which is exactly the amortized decay ratio in the symmetric case. Denote by $X$ the fixed point $X=f(X)$ at $d=D$ and $\gamma=\Gamma(\beta)$. Our main discovery is that at the uniqueness threshold $\gamma=\Gamma(\beta)$, $\alpha(d,x)$ is maximized at $d=D$ and $x=X$ with $\alpha(D,X)=1$. It is in debt to the magic of the potential method to discover such a harmoniously beautiful coincidence between amortized correlation decay and phase transition of uniqueness. The lemma follows by observing that $\alpha(d,x)$ is strictly decreasing with respect to $\gamma$.

The formal proof can be found in the full version of the paper.
\end{proofsketch}

The following lemma bounds the amortized correlation decay with respect to the refined metric of $M$-based depth.

\begin{lemma}\label{lemma-potentialized-decay}
Assume that $0\le\beta<1$ and $\Gamma(\beta)<\gamma<\frac{1}{\beta}$.
There exist constants $\alpha<1$ and $M>1$ which depend only on $\beta$ and $\gamma$, for every vertex $v\in B_M(L)$, assuming that $v$ has $d_0$ children fixed to be blue, $d_1$ children fixed to be green, and $d$ free children $v_1,\ldots,v_d$, it holds that
\begin{align}
\epsilon_v
&\le
M \alpha^{\lceil \log_M (d_0+d_1+d+1) \rceil-1};\label{eq:basis}\\
\epsilon_v
&\le  \alpha^{\lceil \log_M (d_0+d_1+d+1) \rceil} \cdot  \max_{1\le i\le d}\left\{\epsilon_{v_i}\right\}.\label{eq:one-step}
\end{align}
\end{lemma}
\begin{proof}
We choose $\alpha$ to be the one in Lemma \ref{lemma-alpha}, and $M>1$ to satisfy
\begin{align}
\frac{k}{\gamma^{k\frac{D-1}{2D}}}
\leq \alpha^{\lceil \log _Mk \rceil} \quad\text{for $k\ge M$}.\label{eq:M}
\end{align}

Due to \eqref{eq:epsilon-bootstrap},
\begin{align*}
\epsilon_v
&\le
\frac{d}{\gamma^{(d_0+d_1+d)\frac{D-1}{2D}}}
\le  M \alpha^{\lceil \log _M (d_0+d_1+d+1) \rceil-1},
\end{align*}
where the last inequality follows from \eqref{eq:M} if $d_0+d_1+d\geq M$ and the case $d_0+d_1+d < M$ is trivial since $\frac{d}{\gamma^{(d_0+d_1+d)\frac{D-1}{2D}}}<d\leq M$. Thus \eqref{eq:basis} is proved.

Due to \eqref{eq:amortized-decay},
$\epsilon_v\le\alpha(d; \widetilde{R_1},\ldots,\widetilde{R_d})\cdot \max_{1\le i\le d}\{\epsilon_{v_i}\}$
where $R_{v_i}\le\widetilde{R_i}\le R_{v_i}+\delta_{v_i}$. Since $v\in B_M(L)$, its children $v_i\in B_M^*(L)$. As we discussed in the beginning of this section, $0<R_{v_i}\le R_{v_i}+\delta_{v_i}\le 1$, thus $\widetilde{R_i}\in(0,1]$. Then due to Lemma \ref{lemma-alpha}, there is a constant $\alpha<1$,
\begin{align}
\epsilon_v \le \alpha(d; \widetilde{R_1},\ldots,\widetilde{R_d})\cdot \max_{1\le i\le d}\{\epsilon_{v_i}\}\le  \alpha\cdot  \max_{1\le i\le d}\{\epsilon_{v_i}\}. \label{eq:main}
\end{align}
Thus, \eqref{eq:one-step} holds trivially when $d_0+d_1+d< M$.
As for $d_0+d_1+d\ge M$, due to \eqref{eq:step-1},
\begin{align*}
\epsilon_v\le \frac{d}{\gamma^{(d_0+d_1+d)\frac{D-1}{2D}}}\cdot \max_i\{\epsilon_{v_i}\}
\le \alpha^{\lceil \log_M (d_0+d_1+d+1) \rceil} \cdot  \max_{1\le i\le d}\left\{\epsilon_{v_i}\right\}.
\end{align*}
Therefore, \eqref{eq:one-step} is proved.
\end{proof}

\noindent\textbf{Proof of Theorem \ref{thm:rank-decay}.}
%
%
We prove by structural induction in $B_M(L)$. The hypothesis is
 \[\forall v\in B_M(L),\quad \epsilon_v \leq M \alpha^{L-L_M(v)-1}.\]
For the basis, we consider those vertices $v\in B_M(L)$ whose children are in  $B_M^*(L)\setminus B_M(L)$.
The fact that the children of $v$ are not in $B_M(L)$ implies that $L_M(v)+\lceil\log _M (d_0+d_1+d+1)\rceil >L$, where $d_0+d_1+d$ is the number of children of $v$.
Then due to  \eqref{eq:basis} of Lemma \ref{lemma-potentialized-decay},
\begin{align*}
\epsilon_v \le M \alpha^{\lceil \log _M (d_0+d_1+d+1) \rceil-1} \le M \alpha^{L-L_M(v)-1}.
\end{align*}
The induction step is straightforward.
For every child $v_i$ of $v$, $L_M(v_i)=L_M(v)+\lceil \log_M (d_0+d_1+d+1) \rceil$. Suppose that the induction hypothesis is true for all $v_i$. Due to \eqref{eq:one-step} of Lemma \ref{lemma-potentialized-decay}, 
\begin{align*}
\epsilon_v
\le  \alpha^{\lceil \log_M (d_0+d_1+d+1) \rceil} \cdot  \max_{1\le i\le d}\left\{\epsilon_{v_i}\right\}
\le  \alpha^{\lceil \log_M (d_0+d_1+d+1) \rceil+L-L_M(v_i)-1}
= \alpha^{L-L_M(v)-1}.
\end{align*}
The hypothesis is proved. 

Finally, for the root $r$ of the tree,  $L_M(r)=0$, thus due to \eqref{eq:epsilon-delta}, there exists an $\widetilde{R}\in[R_r,R_r+\delta_r]\subseteq (0,1]$ such that 
\begin{align*}
\delta_r=\Phi(\widetilde{R})\cdot\epsilon_r 
\le \widetilde{R}^{\frac{D+1}{2D}}(\beta\widetilde{R}+1)\cdot M \alpha^{L-L_M(r)-1}
\le 2M \alpha^{L-1}.
\end{align*}
As we discussed in the beginning of this section, this implies Theorem \ref{thm:rank-decay}.


\ifabs{}{
\section{A tight analysis for $\beta=0$}\label{sec-beta=0}
In this section, we give a slightly improved and tight analysis (since we also have a hardness result) of the algorithm
when $\beta=0$. In the definition of $\Gamma(\beta)$, we take the maximum over all the possible real $d\geq 1$. As degrees of graphs, only those integer values have physical meanings and we also believe that the maximum value over all the integer $d$ gives the right boundary between tractable and hard. 
In the following, we show how to extend our result to integral $d$ for the special case of $\beta=0$.

Recall that $0\le\beta<1$, $\hat{x}$ satisfies $\hat{x} =\left(\frac{\beta \hat{x}+1}{\hat{x}+\gamma}\right)^d$. The integer version of $\Gamma(\beta)$ can be formally defined as
 $$\Gamma^*(\beta) =\inf\left\{\gamma\ge1 \Bigm{|} \forall d\in \{1,2,3, \ldots\}, \, \frac{d(1-\beta\gamma)(\beta \hat{x}+1)^{d-1}}{(\hat{x}+\gamma)^{d+1}}\le1\right\}.$$

For $\beta=0$, we can solve it and have that $\Gamma^*(0)=\max_{d\in \{1,2,3, \ldots\}} (d-1) d^{-\frac{d}{d+1}}$ .
It is easy to verify that $(d-1) d^{-\frac{d}{d+1}}$ is monotonously increasing when $d\leq 11$, decreasing when $d\geq 12$ and reaching the maximum when $d=11$. Therefore $\Gamma^*(0)=10 \cdot 11^{-\frac{11}{12}} $.

We notice that $\Gamma^*(0)=10 \cdot 11^{-\frac{11}{12}} \approx1.1101714$ and the continuous version $\Gamma(0)\approx1.1101715$. The integrality gap is almost negligible, especially when compared to the previous best boundary for $\gamma$ when $\beta=0$ provided by the heat-bath random walk algorithm in \cite{GJP03}, which is approximately 1.32.

\begin{theorem}\label{thm:alg-beta=0}
Let $A=\begin{bmatrix} 0 & 1 \\ 1 & \gamma \end{bmatrix}$, where $\gamma>\Gamma^*(0)=10 \cdot 11^{-\frac{11}{12}} $.  There is an FPTAS for $Z_A(G)$.
\end{theorem}
\begin{proof}
The algorithm is exactly the same as the algorithm in Section \ref{section-algorithms}. What we need is to establish a correlation decay. For this, we use a special potential function by
substituting $\beta=0$ and $D$ with $D^*=11$. Therefore the potential function is
\[\Phi(R)=R^{\frac{D^*+1}{2D^*}}=R^{\frac{6}{11}}.\]
The analysis remain the same as before, except Lemma \ref{lemma-symmetric-extreme}, which is the only place assuming continuous $d$ in the old analysis. We need to reprove Lemma \ref{lemma-symmetric-extreme} for integral $d$. The symmetric amortized decay $\alpha(d,x)$ is now written as
\[
\alpha^*(d,x)=\frac{d x^{\frac{6}{11}}}{(x+\gamma)^{1+\frac{5d}{11}}}.
\]
We are about to show that if $\gamma>\Gamma^*(0)$, there is a constant $\alpha<1$ such that $\alpha^*(d,x) \le \alpha<1$ for all $0\le x<1$.
Also by the strict monotonicity, we only need to prove (by substituting $\gamma$ with $\Gamma^*(0)$)
\[\alpha^*(d,x)= \frac{d x^{\frac{6}{11}}}{
(x+\Gamma^*(0))^{1+\frac{5d}{11}}} \le 1.\]
%
%
Take the partial derivative of  $\alpha^*$ over $x$, we have
\[\pderiv{\alpha^*}{x}=-\frac{d}{11 x^{\frac{5}{11}}(x+\Gamma^*(0))^{2+\frac{5d}{11}}} ((5+5 d)x -6\Gamma^*(0)).\]
For a fixed $d$, when $x< \frac{6 \Gamma^*(0)}{5+5 d}$, $\alpha^*(d,x)$ is monotonous increasing with $x$ and when
 $x> \frac{6 \Gamma^*(0)}{5+5 d}$, $\alpha^*(d,x)$ is monotonous decreasing with $x$. So  $\alpha^*(d,x)$ reach its
maximum when $x=\frac{6 \Gamma^*(0) }{5+5 d}$. Substituting this into   $\alpha^*(d,x)$, we have
\[\alpha^*(d,x) \leq \hat{\alpha}(d)= \frac{2^{\frac{1}{11}} 3^{\frac{6}{11}} 11^{\frac{5(1+d)}{12}} d (1+d)^{\frac{5}{11}}}{(11+5 d) (10+\frac{12}{1+d})^{\frac{5 d}{11}}}.\]
We can verify that $\hat{\alpha}(d)$ is monotonously increasing when $d\leq 11$ and decreasing when $d\geq 12$ and it reach its maximum when $d=11$. The maximum is $\hat{\alpha}(11)=1$. This completes the proof.
\end{proof}

For $\beta=0$, it is very related to the hardcore model. We can make use of the hardness result in \cite{Sly10} and \cite{GGSVY11} to get a tight hardness result as follows.

\begin{theorem}\label{thm:hard-beta=0}
Let $A=\begin{bmatrix} 0 & 1 \\ 1 & \gamma \end{bmatrix}$, where $\gamma<\Gamma^*(0)= 10 \cdot 11^{-\frac{11}{12}} $.  There is no
FPRAS for $Z_A(G)$ unless $NP=RP$.
\end{theorem}
\begin{proof}
The starting point is the hardness result for hardcore model in~\cite{Sly10}. For hardcore model, the partition function is
\[Z_\lambda(G)=\sum_{S \in I(G)}\lambda^{|S|},\]
where the summation goes over all the independent set of $G$. For $\beta=0$, nonzero terms in the summation
\[ Z_A(G)=\sum_{\sigma \in 2^{V}} \prod_{(i,j) \in E} A_{\sigma(i), \sigma(j)}\]
have a one-to-one corresponding with all the independent sets of $G$. The term indexed by $\sigma$ is nonzero iff $\sigma^{-1}(0)$
is an independent set of $G$.
So $Z_A(G)$ can be rewritten as
\[Z_A(G)= \gamma^{|E|} \sum_{S \in I(G)} \prod_{v\in S} \gamma^{-d(v)}, \]
where $d(v)$ is the degree of vertex $v$. If $G$ is a $d$-regular graph, this summation can be further rewritten as
\[Z_A(G)= \gamma^{|E|} \sum_{S \in I(G)}  (\gamma^{-d})^{|S|}. \]
Since $\gamma^{|E|}$ is a global factor which can be easily computed, the computation for $Z_A(G)$ of $d$-regular graph $G$ is equivalent to the partition
function of the hardcore model on $G$ with fugacity parameter $\gamma^{-d}$.
In \cite{Sly10} and \cite{GGSVY11}, it is proved that there is no FPRAS for the partition function for hardcore model on graphs
with maximum degree $d$ when the  fugacity parameter $\lambda>\frac{(d-1)^{d-1}}{(d-2)^d}$ unless NP$=$RP,
when $d\geq 6$. If we can strength the hardness result to $d$-regular graph, we can use the equivalence relation to get
a hardness result for the the two-spin system model when $\beta=0$ and $\gamma^{-d}>\frac{(d-1)^{d-1}}{(d-2)^d}$.
Let $d=12$, the inequality gives $\gamma<10 \cdot 11^{-\frac{11}{12}}$, as what we claimed.
In the following, we show
that their hardness proof for hardcore model indeed already works for $d$-regular graph.

To prove the hardness of the hardcore model. A reduction from the max-cut problem to the hardcore partition function is built in \cite{Sly10}. The hard instance of the hardcore problem in their reduction is almost $d$-regular
 except some vertices with degree $d-1$. It can be easily verified in their gadget that if we are starting from a max-cut instant in a regular graph, we can choose the suitable parameter and build the reduction to a $d$-regular instance in the hardcore model. So it remains to show that max-cut on a regular graph is already NP-hard.

 This can be done by a simple reduction from max-cut on arbitrary graph to a max-cut instance of a regular graph. Let $G=(V,E)$ be a given max-cut instance. Let $\Delta$ be the maximum degree of $G$.  Then the new instance is of $2\Delta$-regular.  The new graph $G'=(V',E')$ is defined as follows:
\begin{itemize}
  \item  For every vertex $v\in V$, we construct $1+2(\Delta-d(v))$ vertices in $V'$, we name them as $v$ and  $v_i^{+}, v_i^{-}$ for $i=1,2,\ldots, \Delta-d(v)$. These are all the vertices in the new graph $G'$.
  \item For every $v\in V$ and $i\in \{1,2,\ldots, \Delta-d(v)\}$, we connect $2 \Delta -1$ edges in $G'$ between
  $v_i^{+}$ and $v_i^{-}$, one edge between $v$ and $v_i^{+}$, and one edge between $v$ and $v_i^{-}$.
  \item For  every $(u,v) \in E$ be an edge of $E$, we connect two edges between $u$ and $v$ in $G'$.
\end{itemize}
It is easy to see that all the vertices in graph $G'$ have degree  $2\Delta$. For a max-cut for $G'$, we will always put
$v_i^{+}$ and $v_i^{-}$ into different sides for every $v\in V$ and $i\in \{1,2,\ldots, \Delta-d(v)$. If not, one can improve the cut by moving one of them to the other side. Given that $v_i^{+}$ and $v_i^{-}$ are always in different sides, the contribution in the cut for the edges between $v_i^{+}$ and $v_i^{-}$, $v$ and $v_i^{+}$, $v$ and $v_i^{-}$ are all fixed. The remaining part is identical to the original graph except that we double every edge. This finishes the reduction and completes the proof. 
\end{proof}
}

\section{Open Questions}
Our analysis of correlation decay assumes a continuous degree $d$ because of the the using of differentiation. An open question is to improve the analysis to integral $d$ and the uniqueness threshold realized by infinite $(d+1)$-regular trees $\widehat{\mathbb{T}}^d$. It will be very interesting to prove a hardness result beyond this threshold and observe the similar transition of computational complexity in spin systems as in the hardcore model~\cite{Sly10}.

In this paper, we consider the two-state spin systems without external fields. It will be interesting to extend our result to cases where there is an external field as in \cite{GJP03}. Since the hardcore model can be expressed as a two-state spin system with an external field. This will give a unified theory covering the previous results for the hardcore model.

Most importantly, it will be interesting to apply the general technique in this paper to design FPTAS for other counting problems.





\end{document}